\documentclass[a4paper,UKenglish,cleveref, autoref, thm-restate]{lipics-v2021}

\usepackage[utf8]{inputenc}
\usepackage{amssymb}
\usepackage{amsfonts}
\usepackage{mathtools}
\usepackage{multicol}
\usepackage{hyperref}
\usepackage{graphicx}
\usepackage{amsmath}
\usepackage{stmaryrd}

\usepackage{tikz}
\usepackage{pgfplots}
\usetikzlibrary{decorations.markings}
\usetikzlibrary{matrix,backgrounds,folding}
\usetikzlibrary{shapes.geometric, shapes.misc}
\pgfdeclarelayer{edgelayer}
\pgfdeclarelayer{nodelayer}
\pgfsetlayers{background,edgelayer,nodelayer,main}
\tikzset{every path/.style={draw=black!80, line width=0.6pt}}
\tikzstyle{every picture}=[baseline=-0.25em]
\tikzstyle{none}=[inner sep=0mm]
\tikzstyle{black dot}=[inner sep=0.7mm,minimum width=0pt,minimum height=0pt,fill=black,draw=black,shape=circle]
\tikzstyle{dot}=[black dot]
\tikzstyle{white dot}=[dot,fill=white]
\tikzstyle{box}=[rectangle,fill=white,draw=black, font=\scriptsize, inner sep=2pt]
\tikzstyle{arrow}=[decoration={markings,mark=at position 1 with
    {\arrow[scale=1.2,>=stealth]{>}}},postaction={decorate}]
\tikzstyle{box-no-outline}=[rectangle, draw=white, fill=white, inner sep=2pt]
\tikzstyle{polynomial}=[regular polygon, regular polygon sides=3, shape border rotate= 180, fill=white,draw=black,inner sep = -10pt, rounded corners, rounded corners=10pt, minimum width=25pt, font=\scriptsize]

\pgfdeclarelayer{edgelayer}
\pgfdeclarelayer{nodelayer}
\pgfsetlayers{background,edgelayer,nodelayer,main}
\tikzstyle{none}=[inner sep=0mm]
\tikzstyle{every loop}=[]

\usepackage{etoolbox}

\newtoggle{extern}
\toggletrue{extern}
\togglefalse{extern}

\newcommand{\tikzfig}[1]{
\input{./figures/#1.tikz}
}

\def\fig{}

\iftoggle{extern}{
	\usetikzlibrary{external}
	\tikzexternalize[prefix=tikzfigs/]
	
	\renewcommand{\tikzfig}[1]{
		\tikzsetnextfilename{#1}
		
\input{./figures/#1.tikz}
}
	
}{
	
}

\newcommand{\eq}[2][~~]{
#1
\underset{\substack{#2}}{=}
#1
}

\newcommand{\rewrite}[1]{
\xrightarrow[\substack{#1}]{}
}
\newcommand{\interp}[1]{\left\llbracket #1 \right\rrbracket}
\newcommand{\abs}[1]{\left\lvert #1 \right\rvert}
\newcommand{\bra}[1]{\ensuremath{\left\langle \textstyle #1 \right|}}
\newcommand{\ket}[1]{\ensuremath{\left| \textstyle #1 \right\rangle}}
\newcommand{\ketbra}[2]{\ket{#1}\!\!\bra{#2}}

\newcommand{\cat}[1]{\mathbf{#1}}
\newcommand{\atan}[1]{\operatorname{arctan}\left(#1\right)}
\newcommand{\Var}{\operatorname{Var}}

\newlength{\delimstretch}

\newcommand{\ascend}[1]{
\setlength{\delimstretch}{\heightof{\ensuremath{#1}}}
\pgfmathsetlength{\delimstretch}{max(\delimstretch,1.5ex)}
\resizebox{!}{1.1\delimstretch}{$\raisebox{0.18ex}{$\upharpoonleft$}\hspace*{-0.94ex}\lfloor$}\ensuremath{#1}\resizebox{!}{1.1\delimstretch}{$\rfloor\hspace*{-0.94ex}\raisebox{0.18ex}{$\upharpoonright$}$}
}

\newcommand{\descend}[1]{
\setlength{\delimstretch}{\heightof{\ensuremath{#1}}}
\pgfmathsetlength{\delimstretch}{max(\delimstretch,1.5ex)}
\resizebox{!}{1.1\delimstretch}{$\raisebox{-0.18ex}{$\downharpoonleft$}\hspace*{-0.94ex}\lceil$}\ensuremath{#1}\resizebox{!}{1.1\delimstretch}{$\rceil\hspace*{-0.94ex}\raisebox{-0.18ex}{$\downharpoonright$}$}
}

\allowdisplaybreaks
\nolinenumbers

\title{Completeness of Sum-Over-Paths for Toffoli-Hadamard and the Dyadic Fragments of Quantum Computation}

\author{Renaud Vilmart}{Université Paris-Saclay, CNRS, ENS Paris-Saclay, Inria, LMF, 91190, Gif-sur-Yvette, France  \and \url{https://rvilmart.github.io/}}{vilmart@lsv.fr}{https://orcid.org/0000-0002-8828-4671}{}

\authorrunning{R. Vilmart} 

\titlerunning{Completeness of SOP for Tof-H and Dyadic Fragments of Quantum Computation}

\Copyright{Renaud Vilmart} 

\ccsdesc[500]{Theory of computation~Quantum computation theory}
\ccsdesc[300]{Theory of computation~Equational logic and rewriting}

\keywords{Quantum Computation, Verification, Sum-Over-Paths, Rewrite Strategy, Toffoli-Hadamard, Completeness}

\begin{document}

\maketitle

\begin{abstract}
The ``Sum-Over-Paths'' formalism is a way to symbolically manipulate linear maps that describe quantum systems, and is a tool that is used in formal verification of such systems.

We give here a new set of rewrite rules for the formalism, and show that it is complete for ``Toffoli-Hadamard'', the simplest approximately universal fragment of quantum mechanics. We show that the rewriting is terminating, but not confluent (which is expected from the universality of the fragment). We do so using the connection between Sum-over-Paths and graphical language ZH-Calculus, and also show how the axiomatisation translates into the latter.

Finally, we show how to enrich the rewrite system to reach completeness for the dyadic fragments of quantum computation -- obtained by adding phase gates with dyadic multiples of $\pi$ to the Toffoli-Hadamard gate-set -- used in particular in the Quantum Fourier Transform.

\end{abstract}

\section{Introduction}

Sum-Over-Paths (SOP) is a formalism used to represent and manipulate quantum processes in a symbolic way, introduced in 2018 by Amy \cite{SOP}. Its first important feature is its capacity to translate from most common descriptions of quantum processes in polynomial time and space. The formalism hence provides a intermediary view between usual (matrix) semantics and these usual process descriptions. Its second crucial feature is that it comes equipped with a rewrite system that simplifies the term, without altering its semantics.

Despite its links \cite{MSc.Lemonnier,LvdWK} with graphical languages such as the ZH-Calculus \cite{ZH} -- which will be used in the following --, it provides a different view on the quantum processes, representing them as weighted sums of Dirac kets and bras (a very familiar notation in quantum mechanics).

The formalism has seen several applications, the first of which being verification. Verification is a crucial aspect of computations in the quantum realm, where physical constraints (like no-cloning, or the fundamental probabilistic nature of quantum) make it impossible to do debugging the way we do on classical algorithms. More specifically, the SOP formalism was introduced as a solution to circuit equivalence: To check the equivalence between two circuits $\mathcal C_1$ and $\mathcal C_2$, the system represents $\mathcal C_2^\dagger \circ\mathcal C_1$ as an SOP term (where $\mathcal C_2^\dagger$ can be seen as the inverse of $\mathcal C_2$, easy to describe from it). It then tries to reduce it to the identity. If successful, this shows $\mathcal{C}_1$ and $\mathcal C_2$ to represent the same unitary. Otherwise, the system searches for a witness that the term at hand does not represent the identity. As such, the system has been used in several different projects (e.g.~\cite{beaudrap2020fast,Kissinger2020reducing}) to check precisely for circuit equivalence. It was later extended to account for families of morphisms and used within environment Qbricks \cite{chareton2020deductive,chareton2021survey} together with automated solvers to verify algorithms and routines such as quantum phase estimation, Grover's search and Shor's algorithm.

Amongst other applications of the Sum-Over-Paths, we may cite noiseless simulation of quantum processes, where the rewrite strategy is used to reduce the number of variables in the term, effectively decreasing the number of summands when expanding the term to actually compute its semantics. It is for instance one of the simulators implemented in the supercomputer Atos QLM \cite{Haidar2022qlm}.

While the initial suggestion for Sum-Over-Paths focussed on the Clifford+T fragment -- a universal fragment of quantum computing, i.e.~a restriction still capable of approximating with arbitrary precision any quantum process --, it also provided some interesting result for the Clifford fragment. It is known that the latter is not universal \cite{clifford-not-universal}, and actually efficiently simulable with a classical computer, so it is a good test for the relevance of a formalism to check how it handles them. And indeed, it was shown \cite{SOP} a ``weak'' form of confluence of the rewrite system in the Clifford fragment. More precisely, in this fragment, $\mathcal C_2^\dagger \circ\mathcal C_1$ reduces (in polynomial time) to the identity if and only if $\mathcal C_2$ and $\mathcal C_1$ represent the same unitary operator.

However, SOP terms may represent more than unitary operator, but actually any linear map. With those, it is still possible to define the above restrictions, and the rewrite system was extended in \cite{SOP-Clifford} to get confluence for the -- not necessarily unitary -- Clifford fragment. When moving to a universal fragment -- like Clifford+T -- it is expected that we cannot provide a rewrite system with all the good properties of the Clifford case: either reduction is not polynomial, or there is no confluence, or we need an infinite number of rewrites, ... The reason for this is that if we could provide such a system, circuit equivalence would become polynomial, while we know that it is QMA-complete -- a quantum variant of NP-complete -- \cite{Bookatz2014qma,Janzing2005identity}. A weaker property than that of confluence we can ask for is completeness: the question here is to decide whether two equivalent terms can be turned into one another, \emph{with the assumption that rewrites can be used in both directions} (in that case, we rather speak of an equational theory, or axiomatisation, than a rewrite system).

\medskip\noindent
\textbf{Contributions.~}
In this paper, we address the problem of completeness first for arguably the simplest universal fragment of quantum computing, which is \emph{Toffoli-Hadamard}. We provide a fairly simple rewrite system that we show complete for the fragment, and also exhibit two important drawbacks: the non-confluence of the rewrite strategy and the potential explosion of the size of the morphisms during the rewrite. We then show how the rewrite strategy can be tweaked to reach completeness for every dyadic fragment -- where we allow phase gates with phase a multiple of $\frac\pi{2^k}$ for some $k$ --, a restriction that encompasses Clifford, Clifford+T and Toffoli-Hadamard, and is crucially used in the Quantum Fourier Transform, a central block for algorithms such as Shor's and Quantum Phase Estimation.

\medskip\noindent
\textbf{Structure of the paper.~}
After reviewing the Sum-Over-Paths formalism in \autoref{sec:SOP}, the ZH-Calculus in \autoref{sec:ZH} and the links between the two in \autoref{sec:SOP<->ZH}, we show the completeness result for the Toffoli-Hadamard fragment in \autoref{sec:TH-completeness}. The extension to the dyadic fragments is then handled in \autoref{sec:completeness}. 

\noindent
The missing proofs can be found in the appendix.

\section{Sums-Over-Paths}
\label{sec:SOP}

\subsection{The Morphisms}

Sums-Over-Paths \cite{SOP} are a way to symbolically describe linear maps of dimensions powers of $2$ over the complex numbers. These linear maps form a $\dagger$-compact monoidal category \cite{mac2013categories,selinger2010survey} denoted $\cat{Qubit}$ where the objects are natural numbers (this makes the category a PROP \cite{Lack-PROP,PhD.Zanasi}), where morphisms from $n$ to $m$ are linear maps $\mathbb C^{2^n}\to \mathbb C^{2^m}$, and where $(\cdot\circ\cdot)$ (resp.~$(\cdot\otimes\cdot)$) is the usual composition (resp.~tensor product) of linear maps. The category is endowed with a \emph{symmetric braiding} $\sigma_{n,m}:n+m\to m+n$, as well as a \emph{compact structure} $(\eta_n:0\to2n,\epsilon_n:2n\to0)$. Furthermore, there exists an inductive contravariant endofunctor $(\cdot)^\dagger$, that behaves properly with the symmetric braiding and the compact structure. For more information on these structures, see \cite{selinger2010survey}.

The formalism of SOP relies heavily on the Dirac notation for quantum states and operators of $\cat{Qubit}$. The two canonical states of a single qubit are denoted $\ket0$ and $\ket1$. They form a basis of $\mathbb C^2$, and can be viewed as vectors $\ket0=\begin{pmatrix}1&0\end{pmatrix}^\intercal$ and $\ket1=\begin{pmatrix}0&1\end{pmatrix}^\intercal$. A 1-qubit state is then merely a normalised linear combination of these two elements. Using $(\cdot\otimes\cdot)$, they can be used to build the basis states of larger systems, e.g.~$\ket{010}:=\ket0\otimes\ket1\otimes\ket0$ is a basis state of a 3-qubit system. Again, the state of an arbitrary $n$-qubit system is a normalised linear combination of the $2^n$ basis states. We will use extensively the following notation $\bra x$ to represent the dagger (transpose conjugate) of $\ket x$. The identity on a qubit can then be expressed in Dirac notation as $\mathbb I:=\ketbra00+\ketbra11$, where $\ketbra xy := \ket x \circ \bra y$.

We give in the following the definition of Sum-Over-Paths of \cite{SOP-Clifford}, which differs from \cite{SOP} in the way the input qubits are treated, by making them more symmetric with the outputs. This makes some concepts, like the $\dagger$ or the compact structure, more natural.

\begin{definition}[$\cat{SOP}$]
We define $\cat{SOP}$ as the collection of objects $\mathbb N$ and morphisms between them that are tuples $f:n\to m := (s, \vec y, P, \vec O, \vec I)$, which we write:
$\displaystyle s\sum\limits_{\vec y\in V^k} e^{2i\pi P(\vec y)}\ketbra{\vec O(\vec y)}{\vec I(\vec y)}$\\
where $s\in\mathbb{R}$, $\vec y\in V^k$ with $V$ a set of variables, $P\in \mathbb R[X_1,\ldots,X_{k}]/(X_i^2-X_i)$ is called the \emph{phase polynomial} of $f$\footnote{The quotient in the phase polynomial means that we consider each occurrence of the square of a variable to be equal to the variable itself $X_i^2-X_i=0$, since they will be evaluated over $\{0,1\}$. We can further constrain the polynomial by taking it modulo $1$, but only when considered as an element of a group, once all the products have been evaluated, as otherwise all phase polynomials would be evaluated to $0$ as $P = P\times1 = P\times 0 = 0$.}, $\vec O \in \left(\mathbb F_2[X_1,\ldots,X_{k}]\right)^m$ and $\vec I \in \left(\mathbb F_2[X_1,\ldots,X_{k}]\right)^n$ -- $\mathbb F_2$ being the binary field, whose only two elements are its additive and multiplicative identities, denoted $0$ and $1$.

\noindent Compositions are obtained as\footnote{To avoid further clutter, we may not specify the variables of polynomials, e.g.~$P_g$ actually stands for $P_g(\vec y_g)$, $\vec O_g$ for $\vec O_g(\vec y_g)$ etc...}:
\begin{itemize}
\item $f\circ g := \frac{s_fs_g}{2^m}\sum\limits_{\substack{\vec y_f,\vec y_g\\\vec y\in V^m}} e^{2i\pi \left(P_g+P_f+\frac{\vec O_g\cdot \vec y+\vec I_f\cdot \vec y}{2}\right)}\ketbra{\vec O_f}{\vec I_g}$ where $m=\abs{\vec I_f}=\abs{\vec O_g}$
\item $f\otimes g := s_fs_g\sum\limits_{\substack{\vec y_f,\vec y_g}} e^{2i\pi (P_g+P_f)}\ketbra{\vec O_f\vec O_g}{\vec I_f\vec I_g}$
\end{itemize}

\noindent We distinguish particular morphisms:
\begin{itemize}
\item Identity morphisms $id_n : \sum\limits_{\vec y\in V^n}\ketbra{\vec y}{\vec y}$
\item Symmetric braidings $\sigma_{n,m}= \sum\limits_{\vec y_1,\vec y_2}\ketbra{\vec y_2,\vec y_1}{\vec y_1,\vec y_2}$
\item Morphisms for compact structure $\eta_n= \sum\limits_{\vec y} \ketbra{\vec y,\vec y}{}$ and $\epsilon_n= \sum\limits_{\vec y} \ketbra{}{\vec y,\vec y}$
\end{itemize}

\noindent We also distinguish two functors that have $\cat{SOP}$ as a domain:
\begin{itemize}
\item The $\dagger$-functor is given by: $f^\dagger :=  s\sum\limits_{\vec y} e^{-2i\pi P}\ketbra{\vec I}{\vec O}$
\item The functor $\interp{\cdot}:\cat{SOP}\to\cat{Qubit}$ is defined as: $\interp{f}:= s\sum\limits_{\vec y\in \{0,1\}^k} e^{2i\pi P(\vec y)}\ketbra{\vec O(\vec y)}{\vec I(\vec y)}$
\end{itemize}
\end{definition}

The $\dagger$-functor is particularly important to characterise maps that are unitary -- the pure transformations that are allowed by quantum mechanics: $f$ is called unitary if $\interp{f^\dagger\circ f}=id$.

\begin{example}
\label{ex:H-Tof}
The Hadamard and Toffoli gates (which justify the name of the first fragment we will consider in the following), can be represented in this formalism as:
\[H:=\frac1{\sqrt2}\sum_{y_0,y_1}e^{2i\pi\frac{y_0y_1}2}\ketbra{y_1}{y_0}\qquad\operatorname{Tof}:=\sum_{y_0,y_1,y_2}\ketbra{y_0,y_1,y_2\oplus y_0y_1}{y_0,y_1,y_2}\]
It can be checked that both operators are unitary.
\end{example}

As is customary, we consider equality of the SOP morphisms up to $\alpha$-conversion, i.e.~renaming of the variables.
Notice that the definition of the composition $(\cdot\circ\cdot)$ gets somewhat involved. This is to cope with the way we deal with the inputs, which can be any boolean polynomial. The additional terms $\frac{\vec O_g\cdot \vec y+\vec I_f\cdot \vec y}{2}$ enforce that $O_{gi}=I_{fi}$ for all $0\leq i <m$. Indeed, when summing over the variable $y_i$, we get $(1+e^{i\pi(O_{gi}+I_{fi})})$ -- which is non-null only when $O_{gi}=I_{fi}$ -- as a factor of the whole morphism. This presentation has the advantage of keeping the size of the morphism polynomial with the size of the quantum circuit -- or ZH-diagram, see below -- it can be built from, no matter what gate set is used. A downside, however, is that the above does not directly constitute a category, as for instance $id\circ id\neq id$. However, it suffices to quotient the formalism with rewrite rules to turn it into a proper category \cite{SOP-Clifford}, hence justifying the use of the term ``functor'' for the last two maps.

\subsection{A Rewrite System}

We hence give in \autoref{fig:rewrite-TH} a set of rewrite rules denoted $\rewrite{\text{TH}}$ that induces an equational theory $\underset{\textnormal{TH}}{\sim}$ (the symmetric and transitive closure of $\rewrite{\text{TH}}$).

\begin{figure}[!ht]

\[\sum_{\vec y} e^{2i\pi P}\ketbra{\vec O}{\vec I}
\rewrite{y_0\notin \Var(P,\vec O,\vec I)} 2\sum_{\vec y\setminus\{y_0\}} e^{2i\pi P}\ketbra{\vec O}{\vec I}\tag{Elim}\]
\[t = \sum e^{2i\pi\left(\frac{y_0}{2}(y_i\widehat Q + \widehat{Q'} +1) + R\right)}\ket{\vec O}\!\!\bra{\vec I}
\rewrite{y_0\notin\Var(Q,Q',R,\vec O,\vec I)\\y_i\notin\Var(Q,Q')\\QQ'=Q'} t[y_i\leftarrow 1\oplus Q']\tag{HHgen}\]
\[t = \sum e^{2i\pi\left(\frac{y_0}{2}\widehat Q + \frac{y_0'}{2}\widehat{Q'} + R\right)}\ketbra{\vec O}{\vec I}
\rewrite{y_0,y_0'\notin \Var(Q,Q',R,\vec O,\vec I)} 2t[y_0'\leftarrow y_0\oplus y_0Q]\tag{HHnl}\]
\[t=\sum_{\vec y} e^{2i\pi\left(P\right)}|\cdots,\overset{O_i}{\overbrace{y_0{\oplus} O_i'}},\cdots\rangle \!\!\bra{\vec I}
\rewrite{O_i'\neq 0\\y_0\notin \Var(O_1,\ldots,O_{i-1},O_i')}
t[y_0{\leftarrow} O_i]\tag{ket}\]
\vspace*{-0.5em}
\[t=\sum_{\vec y} e^{2i\pi\left(P\right)}\ket{\vec O}\!\!\langle \cdots,\overset{I_i}{\overbrace{y_0{\oplus} I_i'}},\cdots|
\rewrite{I_i'\neq 0\\y_0\notin \Var(\vec O,I_1,\ldots,I_{i-1},I_i')}
t[y_0{\leftarrow} I_i]\tag{bra}\]
\[s\sum_{\vec y} e^{2i\pi\left(\frac{y_0}{2} + R\right)}\ketbra{\vec O}{\vec I}
\rewrite{R\neq 0 \text{ or } \vec O,\vec I\neq \vec 0\\y_0\notin\Var(R,\vec O,\vec I)} \sum_{y_0} e^{2i\pi\left(\frac{y_0}{2}\right)}\ketbra{0,\cdots,0}{0,\cdots,0}\tag{Z}\]
\caption[]{Rewrite system $\rewrite{\text{TH}}$}
\label{fig:rewrite-TH}
\end{figure}

We need in the conditions of all the rules the function $\Var$, that, given a set or list of polynomials, gives the set of all variables used in them.
We call \emph{internal variable} a variable that is present in the morphism $t$ but not in its inputs/outputs, i.e.~a variable $y_0$ such that $y_0\in\Var(t)\setminus\Var(\vec O,\vec I)$. It is worth noting that searching for an occurrence of, and applying any of these rules \emph{once} can be done in polynomial time.

The rules (ket) and (bra) correspond to changes of variables that are necessary to get a unique normal form in the Clifford case \cite{SOP-Clifford}, and the rule (Elim) simply gets rid of a variable that is used nowhere in the term and simply contributes to a global phase (since that variable is supposed to range over two values, it contributes to a multiplicative scalar $2$).

The rules (HHgen), (HHnl) and (Z) all stem from a particular observation: In the morphism $t = \sum e^{2i\pi\left(\frac{y_0}{2}\widehat Q + R\right)}\ketbra{\vec O}{\vec I}$ where $y_0$ is internal and not in $R$, if $Q$ is evaluated to $1$, then the whole morphism is interpreted as null. This is exactly what (Z) captures -- and the conditions on $R$, $\vec O$ and $I$ are simply here to avoid applying the rule indefinitely.

The rule (HHgen) deals with a case where the polynomial $Q$ can be forced to $0$, whilst the rule (HHnl) factorises different such polynomials $Q$ into one.

\begin{remark}
When performing certain rules, we have to substitute a variable by a boolean polynomial $Q$. We need to be able to understand $Q$ as a phase polynomial, as the variable can occur in $P$. The map $\widehat{(\cdot)}:\mathbb F_2[X_1,\ldots,X_{k}]\to \mathbb R[X_1,\ldots,X_{k}]/(X_i^2-X_i)$, serves this purpose. It is inductively defined as:
\[ \widehat{Q_1Q_2}=\widehat{Q_1}\widehat{Q_2}\qquad\qquad \widehat{Q_1\oplus Q_2}=\widehat{Q_1}+\widehat{Q_2}-2\widehat{Q_1}\widehat{Q_2} \qquad\qquad \widehat{y_i}=y_i \qquad\qquad \widehat{\alpha}=\alpha\]
\end{remark}

%

\begin{remark}
When rewriting SOP-morphisms for simplification or verification, it can be beneficial to not only reduce the number of variables -- which is what all rules but (ket/bra) do --, but also to keep the size of the phase polynomial as short as possible. In that respect, the rule (HHgen) can be generalised to:
\[t = \sum e^{2i\pi\left(\frac{y_0}{2}(y_i\widehat Q + \widehat{QQ'} +1) + R\right)}\ketbra{\vec O}{\vec I}
\rewrite{y_0\notin\Var(Q,Q',R,\vec O,\vec I)\\y_i\notin\Var(Q,Q')} t[y_i\leftarrow 1\oplus Q']\tag{HHgen'}\]
where the polynomial $Q'$ can here be smaller (in the number of terms) than the one in (HHgen). However, finding a ``minimal'' $Q'$ for this rule is a hard problem, as it requires the use of boolean Groebner bases \cite{Boolean-Grobner-Bases}. (HHgen) can be seen as a particular case of (HHgen'), where $Q'\leftarrow QQ'$, as $Q\times QQ' = QQ'$. The rule (HHgen) is sufficient for the scope of this paper.
\end{remark}

In \cite{SOP} was introduced a particular and important rule:
\[t = \sum e^{2i\pi\left(\frac{y_0}{2}(y_i + \widehat{Q}) + R\right)}\ketbra{\vec O}{\vec I}
\rewrite{y_0\notin\Var(Q,R,\vec O,\vec I)\\y_i\notin\Var(Q)} 2\sum e^{2i\pi R[y_i\leftarrow \widehat Q]}\left(\ketbra{\vec O}{\vec I}\right)[y_i\leftarrow Q]\tag{HH}\]
This one is a particular case of the rule (HHgen) (with additional use of the rule (Elim)), where $Q\leftarrow 1$, $Q'\leftarrow Q\oplus1$. Moreover, the rule gave enough power to the formalism to become a $\dagger$-compact PROP \cite{SOP-Clifford}. We can extend this result here thanks to:
\begin{proposition}
\label{prop:TH-local}
\[\forall t_1,t_2 \in \cat{SOP},~~t_1\underset{\textnormal{TH}}{\sim}t_2 \implies \begin{cases}
A\circ t_1 \circ B \underset{\textnormal{TH}}{\sim} A\circ t_2 \circ B & \text{for all $A$, $B$ composable}\\
A\otimes t_1 \otimes B \underset{\textnormal{TH}}{\sim} A\otimes t_2 \otimes B & \text{for all $A$, $B$}
\end{cases}\]
\end{proposition}

\begin{proof}
In appendix at page \pageref{prf:TH-local}.
\end{proof}

Thanks to this Proposition, and since $\cat{SOP}/\underset{\textnormal{HH}}{\sim}$ is a $\dagger$-compact PROP by \cite{SOP-Clifford}, we get:
\begin{corollary}
$\cat{SOP}/\underset{\textnormal{TH}}{\sim}$ is a $\dagger$-compact PROP.
\end{corollary}

The set of rules was obviously chosen so as to preserve the semantics:
\begin{proposition}[Soundness]
\label{prop:TH-soundness}
For any two $\cat{SOP}$ morphisms $t_1$ and $t_2$, if $t_1\overset{\ast}{\rewrite{\textnormal{TH}}}t_2$, then $\interp{t_1}=\interp{t_2}$.
\end{proposition}

\begin{proof}
In appendix at page \pageref{prf:TH-soundness}.
\end{proof}

\begin{example}
The following morphism:
\[\sum_{\vec y}e^{2i\pi(\frac{y_0y_1y_2}2+\frac{y_2}2+\frac{y_2y_3y_4}2)}\ketbra{y_4}{y_0}\]
is irreducible using the rules of \cite{SOP} and \cite{SOP-Clifford}. However, here it can be reduced to:
\begin{align*}
\sum_{y_0,y_1,y_2,y_3,y_4}e^{2i\pi(\frac{y_0y_1y_2}2+\frac{y_2}2+\frac{y_2y_3y_4}2)}&\ketbra{y_4}{y_0}\\
\overset{\text{(HHnl)}}{\rewrite{y_3\leftarrow y_1\oplus y_0y_1y_2}}
&~2\sum_{y_0,y_1,y_2,y_4}e^{2i\pi(\frac{y_0y_1y_2}2+\frac{y_2}2+\frac{y_1y_2y_4}2+\frac{y_0y_1y_2y_4}2)}\ketbra{y_4}{y_0}\\
\overset{\text{(HHgen)}}{\rewrite{y_1\leftarrow 1}}
&~2\sum_{y_0,y_2,y_4}e^{2i\pi(\frac{y_0y_2}2+\frac{y_2}2+\frac{y_2y_4}2+\frac{y_0y_2y_4}2)}\ketbra{y_4}{y_0}
\end{align*}
The first rewrite can be made clearer by writing the phase polynomial as $\frac{y_1}2(y_0y_2)+\frac{y_3}2(y_2y_4)+\frac{y_2}2$, and the second one by writing it as $\frac{y_2}2\left(y_1(y_0+y_4+y_0y_4)+0+1\right)$. Recall that variables are binary variables, so $y_i^2=y_i$.
\end{example}

\section{The $\cat{ZH}$-Calculus}
\label{sec:ZH}

The graphical calculi ZX, ZW and ZH \cite{ZH,interacting,ghz-w} are calculi for quantum computing, with a tight link with the Sum-Over-Paths formalism \cite{MSc.Lemonnier,LvdWK,SOP-Clifford}, and whose completeness was proven in particular for the Toffoli-Hadamard fragment \cite{backens2021ZHcompleteness,zw,phase-free-ZH,zx-toffoli}. 

This fragment of quantum mechanics is approximately universal \cite{toffoli-simple,toffoli}, and it is arguably the simplest one with this property. This is the fragment we will be interested in, in most of the following of the paper; and the associated completeness result will be paramount in the development of the following.

We choose to present here the ZH-Calculus, because of its proximity with both $\cat{SOP}$ and the Toffoli-Hadamard fragment. Notice however that there exist translations between all the aforementioned graphical calculi, so by composition, we can connect $\cat{SOP}$ to all of them.


$\cat{ZH}$ is a PROP whose morphisms -- read here from top to bottom -- are composed (sequentially $(\cdot\circ\cdot)$ or in parallel $(\cdot\otimes\cdot)$) from Z-spiders and H-spiders:
\begin{itemize}
\item $Z_m^n:n\to m::$\tikzfig{Z-spider}, called Z-spider
\item $H_m^n(r):n\to m::$\tikzfig{H-spider}, called H-spider, with a parameter $r\in\mathbb C$
\end{itemize}
When $r$ is not specified, the parameter in the H-spider is taken to be $-1$.

$\cat{ZH}$ is made a $\dagger$-compact PROP, which means it also has a symmetric structure $\sigma_{n,m}::\tikzfig{sigma}$, a compact structure $\left(\eta_n::\tikzfig{eta},\epsilon_n::\tikzfig{epsilon}\right)$, and a $\dagger$-functor $(\cdot)^\dagger:\cat{ZH}^{\operatorname{op}}\to\cat{ZH}$. It is defined by: $(Z_m^n)^\dagger := Z_n^m$ and $(H_m^n(r))^\dagger := H_n^m(\overline{r})$ 
where $\overline r$ is the complex conjugate of $r$. 
For convenience, we define two additional spiders:\\
\tikzfig{X-spider} and \tikzfig{X-spider-neg}

The language comes with a way of interpreting the morphisms as morphisms of $\cat{Qubit}$. The standard interpretation $\interp{\cdot}:\cat{ZH}\to\cat{Qubit}$ is a $\dagger$-compact-PROP-functor, defined as:
\[\interp{\tikzfig{Z-spider}} = \ketbra{0^m}{0^n} + \ketbra{1^m}{1^n}
\qquad\qquad\interp{~\tikzfig{id}~}=\ketbra00+\ketbra11\]
\[\interp{\tikzfig{H-spider}} = \sum_{j_k,i_k\in\{0,1\}}r^{j_1\ldots j_mi_1\ldots i_n}\ketbra{j_1,\ldots,j_m}{i_1,\ldots,i_n}\]
\[\interp{\tikzfig{eta}} = \sum_{i_k\in\{0,1\}}\ket{i_1,\ldots,i_n,i_1,\ldots,i_n}=\interp{\tikzfig{epsilon}}^\dagger\]
\[\interp{\tikzfig{sigma}} = \sum_{i_k,j_k\in\{0,1\}}\ketbra{j_1,\ldots,j_m,i_1,\ldots,i_n}{i_1,\ldots,i_n,j_1,\ldots,j_m}\]

Notice that we used the same symbol for two different functors: the two interpretations $\interp{\cdot}:\cat{SOP}\to\cat{Qubit}$ and $\interp{\cdot}:\cat{ZH}\to\cat{Qubit}$. It should be clear from the context which one is to be used.\\
The language is universal: $\forall f\in \cat{Qubit},~\exists D_f\in\cat{ZH},~~\interp{D_f} = f$. In other words, the interpretation $\interp{\cdot}$ is surjective. 

The language comes with an equational theory, which in particular gives the axioms for a $\dagger$-compact PROP. We will not present it here.

We can easily define a restriction of $\cat{ZH}$ that exactly captures the Toffoli-Hadamard fragment of quantum mechanics \cite{backens2021ZHcompleteness,phase-free-ZH}, as the language generated by:
$\left\lbrace\tikzfig{H-spider-minus-1},\tikzfig{Z-spider},\tikzfig{H-scalar-1_sqrt2}\right\rbrace$. Notice that the two black spiders can still be defined if we also define $\tikzfig{H-scalar-1_sqrt2_p}:=\tikzfig{H-scalar-1_sqrt2}^{\otimes p}$. We denote this restriction by $\cat{ZH}_{\operatorname{TH}}$.

This restriction is provided with an equational theory, given in \autoref{fig:ZH-TH-rules}\footnote{The axiomatisation provided here is that of \cite{phase-free-ZH}. It was later simplified in \cite{backens2021ZHcompleteness} in a fragment that is very close to the one we consider, but does \emph{not} contain the scalar $\frac1{\sqrt2}$. As we would rather have this scalar in the language (to properly represent the Hadamard gate), instead of giving a mix of the two axiomatisation, we decided to stick to the first one.}, that makes it complete.

\begin{figure}
~\hfill\tikzfig{ZH-rule-ZS1}\hfill\tikzfig{ZH-rule-ZS2}\hfill\tikzfig{ZH-rule-HS1}\hfill\tikzfig{ZH-rule-HS2}\hfill~\\[1.5em]
\tikzfig{ZH-rule-BA1}\hfill\tikzfig{ZH-rule-BA2}\hfill\tikzfig{ZH-rule-M}\hfill\tikzfig{ZH-rule-O}\\[1.5em]
\phantom{.}
\hfill\tikzfig{ZH-rule-AND}
\hfill\tikzfig{ZH-rule-IV}\hfill\tikzfig{ZH-rule-Z}\hfill~\\[1.5em]
\phantom{.}\hfill\tikzfig{X-spider}\hfill\tikzfig{X-spider-neg}\hfill$\tikzfig{H-scalar-1_sqrt2_p}:=\tikzfig{H-scalar-1_sqrt2}^{\otimes p}$\hfill~
\caption[]{Set of rules $\operatorname{ZH}_{\operatorname{TH}}$ \cite{phase-free-ZH}.}
\label{fig:ZH-TH-rules}
\end{figure}

\begin{theorem}[\cite{phase-free-ZH} Completeness of $\cat{ZH}_{\operatorname{TH}}/\operatorname{ZH}_{\operatorname{TH}}$]
\label{thm:ZH-completeness}
\[\forall D_1,D_2\in \cat{ZH}_{\operatorname{TH}},~~\interp{D_1}=\interp{D_2}\iff \operatorname{ZH}_{\operatorname{TH}}\vdash D_1=D_2\]
\end{theorem}

\section{Translations between $\cat{SOP}$ and $\cat{ZH}$}
\label{sec:SOP<->ZH}

\subsection{From $\cat{SOP}$ to $\cat{ZH}$}

It is possible to translate $\cat{SOP}$ morphisms to ZH-diagrams using interpretation $[\cdot]^{\operatorname{ZH}}:\cat{SOP}\to\cat{ZH}$. 
A description of $[\cdot]^{\operatorname{ZH}}:\cat{SOP}\to\cat{ZH}$ was defined in \cite{MSc.Lemonnier,LvdWK} and in \cite{SOP-Clifford}. We choose the latter definition as it fits our definition of $\cat{SOP}$.
\[\left[s \sum_{\vec y} e^{2i\pi P}\ketbra{O_1,\ldots,O_m}{I_1,\ldots,I_n}\right]^{\operatorname{ZH}}:=\tikzfig{ZH-NF}\]
where the row of Z-spiders represents the variables $y_1,\ldots,y_k$.
Informally:
\begin{itemize}
\item each monomial $\alpha y_{i_1}...y_{i_s}$ in $P$ gives a single H-spider with parameter $e^{i\frac\alpha{2\pi}}$ and connected to the Z-spiders that represent $y_{i_1}$,...,$y_{i_s}$
\item each monomial $y_{i_1}...y_{i_s}$ in $O_i$ is represented by \tikzfig{and-n} where the inputs are connected to the Z-spiders that represent $y_{i_1}$,...,$y_{i_s}$. Notice that the only (non-zero) constant monomial is \tikzfig{constant-monomial}
\item these monomials are then added to form $O_i$ thanks to \tikzfig{xor-n}
\item the nodes $I_i$ are defined similarly, but upside-down
\end{itemize}
For more details, see \cite{SOP-Clifford}.

\begin{example}
The $\cat{SOP}$ morphism:
\[\frac{1}{2\sqrt{2}}\sum\limits_{\vec y}e^{2i\pi\left(\frac14y_0+\frac{1}{2}y_4y_0+\frac{1}{8}y_5y_0y_1+\frac{3}{4}y_1y_2y_3+\frac{1}{2}y_0y_3\right)}\ketbra{0,1{\oplus} y_0{\oplus} y_4y_2,y_5}{y_4,y_5{\oplus} y_2{\oplus} y_3}\]
is mapped to \tikzfig{SOP-to-ZH-example}
\end{example}
The boolean polynomials as defined above are given in their (unique) expanded form. These can easily be shown to be copied through the white node:
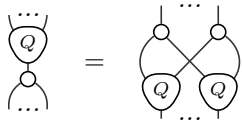
\begin{lemma}
\label{lem:poly-copy}
\def\fig{polynomial-copy}
\begin{align*}
\begin{tikzpicture}
	\begin{pgfonlayer}{nodelayer}
		\node [style=white dot] (0)  at (-0.5, 0.0) {};
		\node [style=white dot] (1)  at (0.5, 0.0) {};
		\node [style=white dot] (2)  at (0.0, 0.5) {};
		\node [style=box] (3)  at (-0.5, 0.5) {};
		\node [style=box] (4)  at (0.0, 0.0) {};
		\node [style=box] (5)  at (0.5, 0.5) {};
		\node [style=none] (6)  at (-0.5, -0.5) {};
		\node [style=none] (7)  at (0.5, -0.5) {};
	\end{pgfonlayer}
	\begin{pgfonlayer}{edgelayer}
		\draw (0) to (6.center);
		\draw (0) to (1);
		\draw (1) to (7.center);
		\draw (2) to (5);
		\draw (2) to (4);
		\draw (3) to (0);
		\draw (3) to (2);
		\draw (5) to (1);
	\end{pgfonlayer}
\end{tikzpicture}
\eq{}\begin{tikzpicture}
	\begin{pgfonlayer}{nodelayer}
		\node [style=none] (53)  at (-0.25, -0.625) {};
		\node [style=none] (54)  at (0.5, -0.625) {};
		\node [style=dot] (55)  at (-0.25, -0.3) {};
		\node [style=none, font={\scriptsize}] (56)  at (-0.5, -0.325) {$\neg$};
		\node [style=dot] (58)  at (0.5, -0.3) {};
		\node [style=none, font={\scriptsize}] (59)  at (0.25, -0.325) {$\neg$};
		\node [style=dot] (60)  at (-0.25, 0.2) {};
		\node [style=none, font={\scriptsize}] (61)  at (-0.5, 0.175) {$\neg$};
		\node [style=dot] (62)  at (0.5, 0.2) {};
		\node [style=none, font={\scriptsize}] (63)  at (0.25, 0.175) {$\neg$};
		\node [style=white dot] (74)  at (0.0, 0.625) {};
	\end{pgfonlayer}
	\begin{pgfonlayer}{edgelayer}
		\draw (54.center) to (58);
		\draw (55) to (53.center);
		\draw (60) to (55);
		\draw (62) to (58);
	\end{pgfonlayer}
\end{tikzpicture}
\end{align*}
\end{lemma}

\begin{proof}
In appendix at page \pageref{prf:poly-copy}.
\end{proof}

\noindent This translation preserves the semantics:
\begin{proposition}[\cite{SOP-Clifford}]
\label{prop:zh-preserves-semantics}
$\interp{[\cdot]^{\operatorname{ZH}}}=\interp{\cdot}$.
\end{proposition}

\subsection{From $\cat{ZH}$ to $\cat{SOP}$}

Any $\cat{ZH}$-diagram can be understood as a $\cat{SOP}$-morphism. To do so, we use the PROP-functor $[\cdot]^{\operatorname{sop}}:\cat{ZH}\to\cat{SOP}$ defined as:
\[\left[~\tikzfig{H-spider-phase}~\right]^{\operatorname{sop}} := \sum e^{2i\pi \frac{\alpha}{2\pi}x_1\ldots x_ny_1\ldots y_m}\ketbra{y_1,\ldots, y_m}{x_1,\ldots,x_n}\]
\[\left[~\tikzfig{H-scalar}~\right]^{\operatorname{sop}} := s\ketbra{}{} \quad\text{ for }s\in\mathbb R\]
\[\left[~\tikzfig{Z-spider}~\right]^{\operatorname{sop}} := \sum_y \ketbra{y,\ldots, y}{y,\ldots, y}\qquad\qquad
\left[~\tikzfig{H-spider-0}~\right]^{\operatorname{sop}} := \left[~\tikzfig{H-spider-0-decomp}~\right]^{\operatorname{sop}}\]
The functor furthermore maps the symmetric braiding (resp. the compact structure) of $\cat{ZH}$ to the symmetric braiding (resp. the compact structure) of $\cat{SOP}$.

This does not give a full description of $[\cdot]^{\operatorname{sop}}$, as we did not describe the interpretation of the H-spider for all parameters, but only for phases and $0$. However, any H-spider can be decomposed using the previous ones:

\begin{lemma}
\label{lem:H-spider-decomp}
For any $r\in\mathbb C$ such that $\abs{r}\notin\{0,1\}$, there exist $s\in\mathbb C$, $\alpha,\beta\in\mathbb R$ such that:
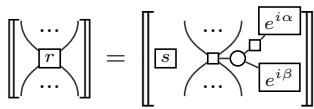
\[\interp{\tikzfig{H-spider}} = \interp{\tikzfig{H-spider-decomp}}\]
\end{lemma}
\begin{proof}
In appendix, at page \pageref{prf:H-spider-decomp}.
\end{proof}

As a consequence, we extend the definition of $[\cdot]^{\operatorname{sop}}$ by:
\[\left[~\tikzfig{H-spider}~\right]^{\operatorname{sop}} :=\left[~\tikzfig{H-spider-decomp}~\right]^{\operatorname{sop}} \]

This interpretation of ZH-diagrams as $\cat{SOP}$-morphisms preserves the semantics:
\begin{proposition}[\cite{SOP-Clifford}]
\label{prop:sop-preserves-semantics}
$\interp{[\cdot]^{\operatorname{sop}}} = \interp{\cdot}$. In other words, the following diagram commutes:
$$\tikzfig{sop-interp-cd}$$
\end{proposition}

The composition of the two interpretations is the identity up to small rewrites:
\begin{proposition}[\cite{SOP-Clifford}]
\label{prop:double-interp-is-identity}
$\left[[\cdot]^{\operatorname{ZH}}\right]^{\operatorname{sop}} \underset{\operatorname{TH}}\sim (\cdot)$
\end{proposition}

\subsection{Restrictions of $\cat{SOP}$}

Recall that $\cat{ZH}_{\operatorname{TH}}$ exactly captures the Toffoli-Hadamard fragment of quantum mechanics. We can then use the two interpretations to define the Toffoli-Hadamard fragment of $\cat{SOP}$. We actually go a step beyond and define a family of fragments indexed by $n$:

\begin{definition}[$\cat{SOP}{[}\frac1{2^n}{]}$]
We define $\cat{SOP}[\frac1{2^n}]$ as the restriction of $\cat{SOP}$ to morphisms of the form: $\displaystyle t = \frac1{\sqrt2^p}\sum e^{2i\pi \frac P{2^n}}\ketbra{\vec O}{\vec I}$ where $p\in\mathbb Z$ and $P$ has integer coefficients.
\end{definition}

The Toffoli-Hadamard fragment is then the first such restriction ($n=1$):
\begin{proposition}
\label{prop:sop1/2-TOF}
$\cat{SOP}[\frac12]$ captures exactly the Toffoli-Hadamard fragment of quantum mechanics.
\end{proposition}
\begin{proof}
We can prove this by showing that $\left[\cat{ZH}_{\operatorname{TH}}\right]^{\operatorname{sop}}\subseteq \cat{SOP}[\frac12]$ and that $\left[\cat{SOP}[\frac12]\right]^{\operatorname{ZH}}\subseteq \cat{ZH}_{\operatorname{TH}}$. The two claims are straightforward verifications, and use the fact that compositions of $\cat{SOP}[\frac12]$-morphisms give $\cat{SOP}[\frac12]$-morphisms.\\
Then, $\interp{\cat{ZH}_{\operatorname{TH}}}=\interp{\left[\cat{ZH}_{\operatorname{TH}}\right]^{\operatorname{sop}}}\subseteq\interp{\cat{SOP}[\frac12]} = \interp{\left[\cat{SOP}[\frac12]\right]^{\operatorname{ZH}}}\subseteq \interp{\cat{ZH}_{\operatorname{TH}}}$, so:
\[\interp{\cat{SOP}[\textstyle\frac12]} = \interp{\cat{ZH}_{\operatorname{TH}}}\qedhere\]
\end{proof}
Notice in particular that the Hadamard and Toffoli gates given in \autoref{ex:H-Tof} lie in this fragment. Not all of $\cat{SOP}[\frac12]$ can be generated by these two gates however, as $\cat{SOP}[\frac12]$ comprises linear maps that are not unitary, i.e.~such that $\interp{t^\dagger\circ t}\neq id$.

\section{Completeness for Toffoli-Hadamard}
\label{sec:TH-completeness}

In this section, we aim to show that the set of rules $\rewrite{\text{TH}}$ captures the whole Toffoli-Hadamard fragment of quantum mechanics. We do so by transporting the similar result from $\cat{ZH}_{\operatorname{TH}}$ to $\cat{SOP}[\frac12]$. First, we show:

\begin{proposition}
\label{prop:TH-proves-TH}
$\forall D_1,D_2\in \cat{ZH}_{\operatorname{TH}},~~ \operatorname{ZH}_{\operatorname{TH}}\vdash D_1=D_2\implies [D_1]^{\operatorname{sop}}\underset{\textnormal{TH}}\sim [D_2]^{\operatorname{sop}}$
\end{proposition}

\begin{proof}
In appendix at page \pageref{prf:TH-proves-TH}.
\end{proof}

We can then use the previous proposition to show the main result of this paper:
\begin{theorem}
\label{thm:TH-completeness}
$\cat{SOP}[\frac12]/\underset{\textnormal{TH}}\sim$ is complete, i.e.: 
$\forall t_1,t_2\in\cat{SOP}[{\textstyle\frac12}],~~\interp{t_1}=\interp{t_2}\iff t_1\underset{\textnormal{TH}}\sim t_2$
\end{theorem}

\begin{proof}
Let $t_1$ and $t_2$ be two $\cat{SOP}[\frac12]$-morphisms such that $\interp{t_1}=\interp{t_2}$. By \autoref{prop:zh-preserves-semantics}:
$\interp{[t_1]^{\operatorname{ZH}}} = \interp{[t_2]^{\operatorname{ZH}}}$.\\
By completeness of $\cat{ZH}_{\operatorname{TH}}/\operatorname{ZH}_{\operatorname{TH}}$ (\autoref{thm:ZH-completeness}):
$\operatorname{ZH}_{\operatorname{TH}}\vdash [t_1]^{\operatorname{ZH}}=[t_2]^{\operatorname{ZH}}$\\
Thanks to \autoref{prop:TH-proves-TH}:
$\left[[t_1]^{\operatorname{ZH}}\right]^{\operatorname{sop}} \underset{\text{TH}}\sim \left[[t_2]^{\operatorname{ZH}}\right]^{\operatorname{sop}}$.
Finally, by \autoref{prop:double-interp-is-identity}:
\[t_1\underset{\text{TH}}\sim\left[[t_1]^{\operatorname{ZH}}\right]^{\operatorname{sop}} \underset{\text{TH}}\sim \left[[t_2]^{\operatorname{ZH}}\right]^{\operatorname{sop}}\underset{\text{TH}}\sim t_2\qedhere\]
\end{proof}

The rewrite system is however not sufficient to get to a unique normal form, as:
\begin{lemma}[Non-Confluence]
The rewrite system $\rewrite{\textnormal{TH}}$ is not confluent.
\end{lemma}

\begin{proof}
The $\cat{SOP}[\frac12]$-morphism:
$\displaystyle t = \sum e^{2i\pi\left(\frac{1}{2}y_{0}y_{6} + \frac{1}{2}y_{8}y_{9}y_{6} + \frac{1}{2}y_{4}y_{5}y_{6} + \frac{1}{2}y_{8}y_{9}y_{12}\right)}\ket{y_{0}}$\\
can be reduced to (at least) three different non-reducible morphisms:
\begin{align*}
\bullet~& t
\rewrite{\text{HH}(y_{6},[y_{0}\leftarrow y_{4}y_{5}{\oplus}y_{8}y_{9}])}
2\sum e^{2i\pi\left(\frac{1}{2}y_{8}y_{9}y_{12}\right)}\ket{y_{4}y_{5}{\oplus}y_{8}y_{9}}\\
\bullet~& t
\rewrite{\text{HHnl}(y_{4},y_{8})}
2 \sum e^{2i\pi\left(\frac{1}{2}y_{9}y_{4}y_{5}y_{6} + \frac{1}{2}y_{0}y_{6} + \frac{1}{2}y_{9}y_{4}y_{6} + \frac{1}{2}y_{9}y_{12}y_{4} + \frac{1}{2}y_{4}y_{5}y_{6}y_{9}y_{12} + \frac{1}{2}y_{4}y_{5}y_{6}\right)}\ket{y_{0}}\\
&\hspace*{-2em}\rewrite{\text{HH}(y_{6},[y_{0}\leftarrow y_{9}y_{12}y_{4}y_{5}{\oplus}y_{9}y_{4}{\oplus}y_{9}y_{4}y_{5}{\oplus}y_{4}y_{5}])}
4 \sum e^{2i\pi\left(\frac{1}{2}y_{9}y_{12}y_{4}\right)}\ket{y_{9}y_{12}y_{4}y_{5}{\oplus}y_{9}y_{4}{\oplus}y_{9}y_{4}y_{5}{\oplus}y_{4}y_{5}}\\
\bullet~& t
\rewrite{\text{HHnl}(y_{6},y_{12})}
2 \sum e^{2i\pi\left(\frac{1}{2}y_{0}y_{8}y_{9}y_{6} + \frac{1}{2}y_{0}y_{6} + \frac{1}{2}y_{8}y_{9}y_{6} + \frac{1}{2}y_{4}y_{5}y_{6}y_{8}y_{9} + \frac{1}{2}y_{4}y_{5}y_{6}\right)}\ket{y_{0}}\\
&\rewrite{\text{HHgen}(y_6,[y_0\leftarrow y_4y_5{\oplus}y_8y_9{\oplus}y_4y_5y_8y_9])}
2 \sum e^{2i\pi\left(\frac{1}{2}y_{8}y_{9}y_{6}\right)}\ket{y_4y_5{\oplus}y_8y_9{\oplus}y_4y_5y_8y_9}\qedhere
\end{align*}
\end{proof}
Another important downside is the potential explosion of the size of the phase polynomial:
\begin{lemma}
Applying (HHnl) $k$ times in a row on an SOP morphism with phase polynomial of size $O(k)$ may give a morphism with phase polynomial of size $O(2^k)$.
\end{lemma}
\begin{proof}
For any $k\geq1$ we can define the following term:
\[t_k:=\sum e^{2i\pi \sum\limits_{i=0}^k\frac{y_{i0}}2(y_{i1}+y_{i2}+1)} \]
on which we can apply (HHnl) $k$ times in a row. In that case we end up with:
\[t_k\rewrite{}^k2^k\sum e^{2i\pi (\frac{y_0}2\prod\limits_{i=0}^k(y_{i1}+y_{i2})+1)}\]
While $t_k$ has only $3(k+1)$ terms (each of degree at most $2$) in its phase polynomial, it can rewrite into a morphism with $2^{k+1}+1$ terms (each of degree at most $3$).
\end{proof}
Hence, if one were to perform simplifications with this rewrite system, they ought to give special attention as to where and in which order to apply the rules.

\section{Completeness for the Dyadic Fragment}
\label{sec:completeness}

We show here how we can turn an $\cat{SOP}[\frac1{2^{n+1}}]$-morphism into an $\cat{SOP}[\frac1{2^{n}}]$-morphism in a ``reversible'' manner. This will allow us to extend the completeness result to all the restrictions $\cat{SOP}[\frac1{2^n}]$. This is particularly interesting as the phase gates with dyadic multiples of $\pi$, used in particular in the quantum Fourier transform, belong in these fragments:
\[R_Z\left(p\frac\pi{2^k}\right) := \sum_{y_0}e^{2i\pi\cdot\frac p{2^{k-1}}}\ketbra{y_0}{y_0}\]

\subsection{Ascending the Dyadic Levels}

These transformations between restrictions of $\cat{SOP}$ are more easily defined on $\cat{SOP}$-morphisms of a particular shape, namely, when their phase polynomial is reduced to a single monomial. Because of this, we show how a $\cat{SOP}$-morphism can be turned into a composition of these.
\begin{lemma}
\label{lem:sop-decomp}
Let $P = \sum m_i \in \mathbb R[X_1,\ldots,X_{k}]/(X_i^2-X_i)$, and $t = s\sum e^{2i\pi P} \ketbra{\vec O}{\vec I}$. Then:
\begin{align*}
\left[\begin{array}{c}
\displaystyle\left(s\sum \ketbra{\vec O}{y_0,...,y_k}\right)\circ\hfill~\\
\displaystyle\left(\sum e^{2i\pi m_1}\ketbra{y_0,...,y_k}{y_0,...,y_k}\right)
\circ \ldots \circ
\left(\sum e^{2i\pi m_\ell}\ketbra{y_0,...,y_k}{y_0,...,y_k}\right)\\
\displaystyle\hfill\circ\left(\sum \ketbra{y_0,...,y_k}{\vec I}\right)
\end{array}\right]
\overset\ast{\rewrite{\textnormal{HH}}}t
\end{align*}
\end{lemma}
Notice that this decomposed form is not unique, as different orderings on the monomials of $P$ define different orderings of the compositions. However, this will not matter.

A particular care is sadly needed for the overall scalar. Because of this, we will first focus on a slightly different notion of restriction of $\cat{SOP}$.

\begin{definition}[$\cat{SOP}{[}\frac1{2^n}{]}'$]
We define $\cat{SOP}[\frac1{2^n}]'$ as the restriction of $\cat{SOP}$ to morphisms of the form: $\displaystyle t = \frac1{2^p}\sum e^{2i\pi \frac P{2^n}}\ketbra{\vec O}{\vec I}$ where $P$ has integer coefficients.
\end{definition}

The only difference with $\cat{SOP}[\frac1{2^n}]$ is that the overall scalar is now a power of $\frac12$ and not of $\frac1{\sqrt2}$. There always exists a $\cat{SOP}[\frac1{2^n}]'$-morphism that represents the same linear map as any $\cat{SOP}[\frac1{2^{n}}]$-morphism.

\begin{lemma}
$\interp{\frac1{\sqrt2}\sum\limits_{y_0\in V} e^{2i\pi\left(\frac{1}{8} + \frac{3}{4}y_{0}\right)}} = 1$. Hence:
\[\textstyle\forall t\in \cat{SOP}[\frac1{2^n}],~\exists t'\in \cat{SOP}[\frac1{2^{\max(3,n)}}]',~~\interp{t}=\interp{t'}\]
\end{lemma}
\begin{proof}
If $t\in \cat{SOP}[\frac1{2^n}]$ and $t\notin \cat{SOP}[\frac1{2^n}]'$, then:\\ $t':=t\otimes\left(\frac1{\sqrt2}\sum e^{2i\pi\left(\frac{1}{8} + \frac{3}{4}y_{0}\right)}\right) \in \cat{SOP}[\frac1{2^{\max(3,n)}}]'$ and $\interp{t'}=\interp{t}$.
\end{proof}

We can now define the family of maps that will link the different levels of the ``dyadic levels'':
\begin{definition}
For any $k\geq1$, we define the functor $\ascend{\cdot}_k:\cat{SOP}[\frac1{2^{k+1}}]'\to\cat{SOP}[\frac1{2^k}]'$, first for morphisms $t=s\sum e^{2i\pi \frac\ell{2^{k+1}}y_{i_1}...y_{i_q}} \ketbra{\vec O}{\vec I}$ with phase polynomial of size 0 or 1:
\[t\mapsto
\begin{cases}
\displaystyle s\sum e^{2i\pi \frac{\ell/2}{2^{k}}y_{i_1}...y_{i_q}} \ketbra{\vec O,y'}{\vec I,y'} = t\otimes id&\text{ if } \ell\bmod2 = 0\\[0.5ex]
\displaystyle s\sum e^{2i\pi \frac{y_{i_1}...y_{i_q}}{2^{k}}\left((\ell-1)/2+y'\right)} \ketbra{\vec O,y'}{\vec I,y'{\oplus}y_{i_1}...y_{i_q}}&\text{ if } \ell\bmod2 = 1
\end{cases}\]
The functor is then extended to any $\cat{SOP}[\frac1{2^{k+1}}]'$-morphism by the decomposition of \autoref{lem:sop-decomp} (and given a particular ordering on the monomials of the phase polynomial).
\end{definition}
Since $\ascend{\cdot}_k$ is defined to be a functor, we have $\ascend{\cdot\,\circ\,\cdot}_k = \ascend{\cdot}_k\circ\ascend{\cdot}_k$. We can show that the ordering of the monomials has no real importance. Indeed, suppose $t_1=\sum e^{2i\pi \frac{\ell_1}{2^{k+1}}y_{i_1}...y_{i_q}} \ketbra{\vec y}{\vec y}$ and $t_2=\sum e^{2i\pi \frac{\ell_2}{2^{k+1}}y_{j_1}...y_{j_r}} \ketbra{\vec y}{\vec y}$. Then: $
\ascend{t_1\circ t_2}_k = \ascend{t_2\circ t_1}_k$ 
quite obviously when either $\ell_1\bmod2=0$ or $\ell_2\bmod2=0$, but also when $\ell_1\bmod2=\ell_2\bmod2=1$:
\begin{align*}
\ascend{t_1\circ t_2}_k
\rewrite{\text{HH}} \sum e^{2i\pi \left(
\begin{array}{c}
\scriptstyle \frac{y_{i_1}...y_{i_q}}{2^{k}}\left((\ell_1-1)/2+y'\right)
 +\frac{y_{j_1}...y_{j_r}}{2^{k}}\left((\ell_2-1)/2+y'\right)\\[1ex]
\scriptstyle + \frac{y_{i_1}...y_{i_q}y_{j_1}...y_{j_r}}{2^{k}}(1-2y')
\end{array}
\right)} \raisebox{-2ex}{\hspace*{-4em}$\ketbra{\vec y,y'}{\vec y,y'{\oplus}y_{i_1}...y_{i_q}{\oplus}y_{j_1}...y_{j_r}}$}\\
\underset{\text{HH}}\longleftarrow \ascend{t_2\circ t_1}_k
\end{align*}
Notice however that $\ascend{\cdot}_k$ adds an input and an output, so necessarily $\ascend{\cdot\otimes\cdot}_k\neq\ascend\cdot_k\otimes\ascend\cdot_k$.

The functors $\ascend{\cdot}_k$ map terms with the same semantics to terms with the same semantics:
\begin{proposition}
\label{prop:ascending-encodes}
$\displaystyle\forall t_1,t_2\in\cat{SOP}[\textstyle\frac1{2^{k+1}}]',~~\interp{t_1}=\interp{t_2}\implies \interp{\ascend{t_1}_k} = \interp{\ascend{t_2}_k}$
\end{proposition}

\begin{proof}
In appendix at page \pageref{prf:ascending-encodes}.
\end{proof}

\subsection{Going Back}

We  now show how to reverse the functors $\ascend{\cdot}_k$.

\begin{definition}
For any $k\geq1$, we define the (partial) map $\descend{\cdot}_k:\cat{SOP}[\frac1{2^k}]'\to \cat{SOP}[\frac1{2^{k+1}}]'$ as:
\[\forall t:n+1\to m+1 \in\cat{SOP}[{\textstyle\frac1{2^k}}]',~~\descend{t}_k:= (id_m\otimes \bra0)\circ t \circ (id_n\otimes \sum e^{2i\pi \frac{y_0}{2^{k+1}}}\ket{y_0})\]
\end{definition}
Notice that $\descend{\cdot}_k$ can only be applied on morphisms that have at least one input and one output.

$\descend{\cdot}_k$ reverses the action of $\ascend{\cdot}_k$ (up to some rewrites):
\begin{proposition}
\label{prop:descending-reverses-ascending}
$\descend{\ascend{\cdot}_k}_k \underset{\textnormal{TH}}\sim (\cdot)$ and $t_1\underset{\textnormal{TH}}\sim t_2 \implies \descend{t_1}\underset{\textnormal{TH}}\sim \descend{t_2}$ for any two terms $t_1,t_2$.
\end{proposition}

\begin{proof}
In appendix at page \pageref{prf:descending-reverses-ascending}.
\end{proof}

\subsection{Completeness}

We may now show completeness first for $\cat{SOP}[\frac1{2^{k+1}}]'$ and then tweak the equational theory to extend the result to $\cat{SOP}[\frac1{2^{k+1}}]$.

\begin{theorem}[Completeness of $\cat{SOP}{[}\frac1{2^{k+1}}{]}'/\sim_{\raisebox{-1ex}{\hspace*{-2ex}\textnormal{\scriptsize TH}}}$]
\label{thm:dyadic-completeness-aux}
\[\forall t_1,t_2\in \cat{SOP}[{\textstyle\frac1{2^{k+1}}}]',~~\interp{t_1}=\interp{t_2}\iff t_1\underset{\textnormal{TH}}\sim t_2\]
\end{theorem}

\begin{proof}
Let $t_1,t_2\in \cat{SOP}[{\textstyle\frac1{2^{k+1}}}]'$ such that $\interp{t_1}=\interp{t_2}$. By \autoref{prop:ascending-encodes}:
\[\interp{\ascend{...\ascend{t_1}_k...}_1}=\interp{\ascend{...\ascend{t_2}_k...}_1}\]
Since $\ascend{...\ascend{t_i}_k...}_1 \in \cat{SOP}[\frac12]'\subset \cat{SOP}[\frac12]$, by completeness of this fragment (\autoref{thm:TH-completeness}): 
\[\ascend{...\ascend{t_1}_k...}_1 \underset{\textnormal{TH}}\sim \ascend{...\ascend{t_2}_k...}_1 \]
Finally, by \autoref{prop:descending-reverses-ascending}:
$t_1\underset{\textnormal{TH}}\sim\descend{...\descend{\ascend{...\ascend{t_1}_k...}_1}_1...}_k \underset{\textnormal{TH}}\sim \descend{...\descend{\ascend{...\ascend{t_2}_k...}_1}_1...}_k\underset{\textnormal{TH}}\sim t_2 $.
\end{proof}

This is not entirely satisfactory, as we would like to relate any two morphisms of the same interpretation. However:

\begin{lemma}
\label{lem:empty-intersection-sop-sop'}
If $t_1\in \cat{SOP}[{\textstyle\frac1{2^{k+1}}}]'$ and $t_2\in\cat{SOP}[{\textstyle\frac1{2^{k+1}}}]\setminus \cat{SOP}[{\textstyle\frac1{2^{k+1}}}]'$, then $t_1\underset{\textnormal{TH}}{\nsim}t_2$.
\end{lemma}

\begin{proof}
There is no rule in $\rewrite{\textnormal{TH}}$ that changes the overall scalar from an odd power of $\frac1{\sqrt2}$ to an even one, or vice-versa.
\end{proof}

However, adding a single rule:
\[\sum_{\vec y} e^{2i\pi \left(\frac18+\frac34y_0+R\right)}\ketbra{\vec O}{\vec I} \rewrite{y_0\notin\Var(R,\vec O,\vec I)} \sqrt2\sum_{\vec y\setminus\{y_0\}} e^{2i\pi R}\ketbra{\vec O}{\vec I} \tag{$\sqrt2$}\]
fixes this caveat. This rule can also be recovered from the more general one:
\[\sum_{\vec y} e^{2i\pi\left(\frac{y_0}{4} + \frac{y_0}{2}\widehat{Q} + R\right)}\ketbra{\vec O}{\vec I}
\rewrite{y_0\notin\Var(Q,R,\vec O,\vec I)} \sqrt{2}\sum_{\vec y\setminus{\{y_0\}}} e^{2i\pi\left(\frac{1}{8}-\frac{1}{4}\widehat{Q} + R\right)}\ketbra{\vec O}{\vec I}\tag{$\omega$}\]
which was already used in \cite{SOP,LvdWK,SOP-Clifford} to deal with the Clifford fragment of quantum mechanics.

With this additional rule at hand, we can derive the general completeness theorem:
\begin{theorem}[Completeness of $\cat{SOP}{[}\frac1{2^{k+1}}{]}/\sim_{\raisebox{-1ex}{\hspace*{-2ex}\textnormal{\scriptsize TH'}}}$]
Let us write $\rewrite{\textnormal{TH'}}~:=~\rewrite{\textnormal{TH}}+\{(\sqrt2)\}$. Then: 
$\forall t_1,t_2\in \cat{SOP}[{\textstyle\frac1{2^{k+1}}}],~~\interp{t_1}=\interp{t_2}\iff t_1\underset{\textnormal{TH'}}\sim t_2$
\end{theorem}

\begin{proof}
Let $t_1,t_2\in \cat{SOP}[{\textstyle\frac1{2^{k+1}}}]$ such that $\interp{t_1}=\interp{t_2}$. Let us also write:
\[t_{\sqrt2}:=\frac{1}{\sqrt2} \sum e^{2i\pi\left(\frac{1}{8} + \frac{3}{4}y_{0}\right)}\]
We define $t_i'$ as:
$\qquad t_i':=\begin{cases}
t_i &\text{ if } t_i\in \cat{SOP}[{\textstyle\frac1{2^{k+1}}}]'\\
t_i\otimes t_{\sqrt2} &\text{ if } t_i\notin \cat{SOP}[{\textstyle\frac1{2^{k+1}}}]'
\end{cases}$.\\
It is easy to check that $t_i'\in \cat{SOP}[{\textstyle\frac1{2^{\max(3,k+1)}}}]'$ and that $t_i\underset{\textnormal{TH'}}\sim t_i'$. By \autoref{thm:dyadic-completeness-aux}:
\[t_1 \underset{\textnormal{TH'}}\sim t_1' \underset{\textnormal{TH'}}\sim t_2' \underset{\textnormal{TH'}}\sim t_2\qedhere\]
\end{proof}

We hence have completeness for all dyadic fragments of quantum computation. By taking their union, we can get completeness for the ``whole dyadic fragment''.

\begin{definition}
Let $\cat{SOP}[\mathbb D] := \bigcup\limits_{k=1}^\infty \cat{SOP}[\frac1{2^{k}}]$ be the whole dyadic fragment of quantum computation.
\end{definition}

\begin{corollary}[Completeness of $\cat{SOP}{[}\mathbb D{]}/\sim_{\raisebox{-1ex}{\hspace*{-2ex}\textnormal{\scriptsize TH'}}}$]
\[\forall t_1,t_2\in \cat{SOP}[{\textstyle\mathbb D}],~~\interp{t_1}=\interp{t_2}\iff t_1\underset{\textnormal{TH'}}\sim t_2\]
\end{corollary}

\section{Conclusion and Discussion}

We have given a new rewrite system for the Toffoli-Hadamard fragment of Sums-Over-Paths, and showed the induced equational theory to be complete. We then extended this rewrite strategy by adding a single new rewrite, which we then proved to be complete for the whole dyadic fragment. As expected from the universality of the fragments at hand, we do not get all the nice properties of the rewriting in the Clifford fragment. In particular, we showed that the rewrite strategies given above are not confluent, and that the size of the terms may grow exponentially when rules are applied carelessly. Whether one of the above two drawbacks can be removed by a different rewrite system remains an open question.

Using the translation from $\cat{SOP}$ to $\cat{ZH}$, this time, we can make sense of the $\cat{SOP}$ rewrite rules as graphical ones. We will focus on the two rules that were not present in the previous works on $\cat{SOP}$, namely (HHgen) and (HHnl). Let us start with the latter.

(HHnl) turns an occurrence of $\frac{y_0}2\widehat Q+\frac{y_0'}2\widehat Q'$ into $\frac{y_0}2(\widehat Q+\widehat Q'+\widehat Q\widehat Q')$, when the two variables are linked to nothing else than their respective polynomials $Q$ and $Q'$. The induced $\cat{ZH}$ identity can be derived using its rules:
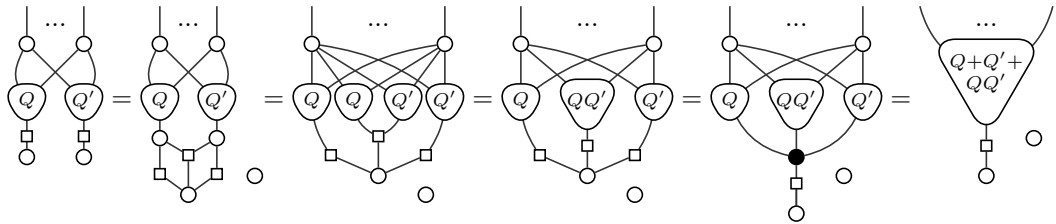
\[\tikzfig{HHnl}\]
(where the first equality uses \autoref{lem:HHnl-lemma} (in appendix), the second, third and last use \autoref{lem:poly-copy}, and the fourth uses (ZS1), (HS2) and the definition of the black node).\\
Although the overall number of nodes usually increases, the number of white nodes that amount to $\cat{SOP}$-variables (i.e. white nodes that are not part of a polynomial) decreases.

Rule (HHgen) is a bit more tricky to deal with in particular as it involves a non-trivial side condition. Hence, we do not provide a derivation of the equality, but only state it. With the pattern $\frac{y_0}2(y_i\widehat{Q}+\widehat{Q}'+1)$ we get $\frac{y_0}2(y_i\widehat{Q}+1)$ with all other occurrences of $y_i$ replaced by $Q'\oplus 1$:
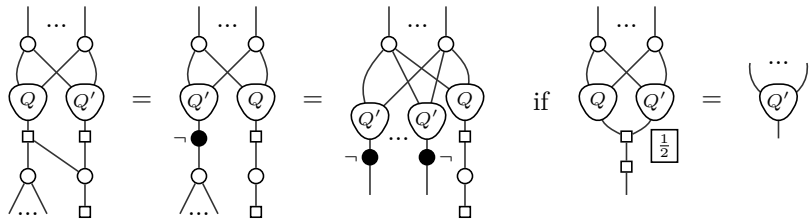
\[\tikzfig{HHgen}\]

This paper, together with the above small study of how the rewrites translate as ZH transformations, really shows how the two formalisms (SOP and ZH) give different and complementary approaches to rewriting and simplifying representations of quantum processes.

We provided new rewrites that allow simplification in the terms -- in that they decrease the number of variables -- with the aim of completeness. A next important step for verification, simulation and simplification using SOP is to determine which rewrites, or which variants, are the most relevant to the task at hand.


\appendix
\section*{Appendix}

\begin{proof}[Proof of \autoref{prop:TH-local}]
\phantomsection\label{prf:TH-local}
The result is obvious for the tensor product $(.\otimes.)$. For the composition, we show that if $t_1\rewrite{\text{TH}}t_2$ in one step, then $A\circ t_1\circ B \underset{\textnormal{TH}}{\sim} A\circ t_2\circ B$. In other words, we have to show it for every rule in $\rewrite{\text{TH}}$:

\noindent$\bullet$ (Elim): Obvious.

\noindent$\bullet$ (HHgen):
\begin{align*}
A\circ t_1\circ B = \sum e^{2i\pi \left(P_A+P_B+\frac{y_0}2(y_i\widehat Q + \widehat{Q'} +1) + R + \frac{\vec O\cdot \vec x + \vec I_A\cdot\vec x + \vec I\cdot \vec x' + \vec O_B\cdot \vec x'}2\right)} \ketbra{\vec O_A}{\vec I_B}\\
\rewrite{\text{HHgen}}
\sum e^{2i\pi \left(P_A+P_B+\frac{y_0}2(\widehat Q +1) + R[y_i\leftarrow \widehat{1{\oplus}Q'}] + \frac{\vec O[y_i\leftarrow \widehat{1{\oplus}Q'}]\cdot \vec x + \vec I_A\cdot\vec x + \vec I[y_i\leftarrow \widehat{1{\oplus}Q'}]\cdot \vec x' + \vec O_B\cdot \vec x'}2\right)} \raisebox{-2ex}{\hspace*{-2em}$\ketbra{\vec O_A}{\vec I_B}$}\\
= A\circ t_1[y_i\leftarrow {1{\oplus}Q'}]\circ B
= A\circ t_2\circ B
\end{align*}

\noindent$\bullet$ (HHnl):
\begin{align*}
A\circ t_1\circ B = \sum e^{2i\pi \left(P_A+P_B+\frac{y_0}{2}\widehat Q + \frac{y_0'}{2}\widehat{Q'} + R + \frac{\vec O\cdot \vec x + \vec I_A\cdot\vec x + \vec I\cdot \vec x' + \vec O_B\cdot \vec x'}2\right)} \ketbra{\vec O_A}{\vec I_B}\\
\rewrite{\text{HHnl}}
2\sum e^{2i\pi \left(P_A+P_B+\frac{y_0}{2}(\widehat Q + \widehat{Q'} + \widehat{QQ'}) + R + \frac{\vec O\cdot \vec x + \vec I_A\cdot\vec x + \vec I\cdot \vec x' + \vec O_B\cdot \vec x'}2\right)} \ketbra{\vec O_A}{\vec I_B}\\
= A\circ (2t_1[y_0'\leftarrow {y_0{\oplus}y_0Q}])\circ B
= A\circ t_2\circ B
\end{align*}

\noindent$\bullet$ (ket):
\begin{align*}
A\circ t_1\circ B =\hspace*{32em}\\
\sum e^{2i\pi \left(P_A+P_B+P+ \frac{(\widehat O_1 +\widehat I_{A1})x_1 +\ldots + (y_0+\widehat O_i'+\widehat I_{Ai})x_i+\ldots +(\widehat O_m +\widehat I_{Am})x_m + \vec I\cdot \vec x' + \vec O_B\cdot \vec x'}2\right)} \ketbra{\vec O_A}{\vec I_B}\\
\rewrite{\text{HH}}
2\sum e^{2i\pi \left(P_A+P_B+P[y_0\leftarrow \widehat{O_i'{\oplus}I_{Ai}}] + \frac{\vec O[y_0\leftarrow \widehat{O_i'{\oplus}I_{Ai}}]\cdot \vec x + \vec I_A\cdot\vec x + \vec I[y_0\leftarrow \widehat{O_i'{\oplus}I_{Ai}}]\cdot \vec x' + \vec O_B\cdot \vec x'}2\right)} \ketbra{\vec O_A}{\vec I_B}\\
\underset{\text{HH}}\longleftarrow
\sum e^{2i\pi \left(P_A+P_B+P[y_0\leftarrow \widehat{y_0{\oplus}O_i'}] + \frac{\vec O[y_0\leftarrow \widehat{y_0{\oplus}O_i'}]\cdot \vec x + \vec I_A\cdot\vec x + \vec I[y_0\leftarrow \widehat{y_0{\oplus}O_i'}]\cdot \vec x' + \vec O_B\cdot \vec x'}2\right)} \ketbra{\vec O_A}{\vec I_B}\\
=A\circ t_1[y_0\leftarrow {y_0{\oplus}O_i'}]\circ B = A\circ t_2\circ B
\end{align*}

\noindent$\bullet$ (bra): Similar to (ket).

\noindent$\bullet$ (Z):
\begin{align*}
A\circ t_1\circ B = \sum e^{2i\pi \left(P_A+P_B+\frac{y_0}2 + R + \frac{\vec O\cdot \vec x + \vec I_A\cdot\vec x + \vec I\cdot \vec x' + \vec O_B\cdot \vec x'}2\right)} \ketbra{\vec O_A}{\vec I_B}
\rewrite{\text{Z}}
\sum e^{2i\pi \left(\frac{y_0}2\right)} \ketbra{\vec 0}{\vec 0}\\
\underset{\text{Z}}\longleftarrow
\sum e^{2i\pi \left(P_A+P_B+\frac{y_0}2 + \frac{\vec 0\cdot \vec x + \vec I_A\cdot\vec x + \vec 0\cdot \vec x' + \vec O_B\cdot \vec x'}2\right)} \ketbra{\vec O_A}{\vec I_B}
= A\circ t_2\circ B
\end{align*}
\end{proof}

\begin{proof}[Proof of \autoref{prop:TH-soundness}]
\phantomsection\label{prf:TH-soundness}
~\\
(HHgen): If $t = \sum_{\vec y\in V^k} e^{2i\pi\left(\frac{y_0}{2}(y_i\widehat Q + \widehat{Q'} +1) + R\right)}\ketbra{\vec O}{\vec I}$ such that $QQ'=Q'$:
\begin{align*}
\interp{t}
&= \sum_{\vec y\in \{0,1\}^k} e^{2i\pi\left(\frac{y_0}{2}(y_i\widehat Q + \widehat{Q'} +1) + R\right)}\ketbra{\vec O}{\vec I}
= \sum_{\vec y\in \{0,1\}^{k-1}} (1+e^{i\pi (y_i\widehat Q + \widehat{Q'} +1)})e^{2i\pi R}\ketbra{\vec O}{\vec I}\\
&= \sum_{\vec y\in \{0,1\}^{k-2}} (1+e^{i\pi (\widehat{Q'}\widehat Q + \widehat{Q'} +1)})e^{2i\pi R[y_i\leftarrow \widehat Q']}(\ketbra{\vec O}{\vec I})[y_i\leftarrow Q']\\[-0.8em]
&\tag*{$\displaystyle + \sum_{\vec y\in \{0,1\}^{k-2}} (1+e^{i\pi (\widehat{1\oplus Q'}\widehat Q + \widehat{Q'} +1)})e^{2i\pi R[y_i\leftarrow \widehat{1\oplus Q'}]}(\ketbra{\vec O}{\vec I})[y_i\leftarrow 1\oplus Q']$}\\
&= 0 +\sum_{\vec y\in \{0,1\}^{k-2}} (1+e^{i\pi (\widehat{1\oplus Q'}\widehat Q + \widehat{Q'} +1)})e^{2i\pi R[y_i\leftarrow \widehat{1\oplus Q'}]}(\ketbra{\vec O}{\vec I})[y_i\leftarrow 1\oplus Q']\\
&=\interp{t[y_i\leftarrow 1\oplus Q']}
\end{align*}
(HHnl) : If $t = \sum e^{2i\pi\left(\frac{y_0}{2}\widehat Q + \frac{y_0'}{2}\widehat{Q'} + R\right)}\ketbra{\vec O}{\vec I}$:
\begin{align*}
\interp{t} &= \sum_{\vec y\in\{0,1\}^k} e^{2i\pi\left(\frac{y_0}{2}\widehat Q + \frac{y_0'}{2}\widehat{Q'} + R\right)}\ketbra{\vec O}{\vec I}\\
&= \sum_{\vec y\in\{0,1\}^{k-2}} \left(1+e^{i\pi\widehat Q}+e^{i\pi\widehat Q'}+e^{i\pi\widehat{Q{\oplus}Q'}}\right) e^{2i\pi R}\ketbra{\vec O}{\vec I}\\
&= \sum_{\vec y\in\{0,1\}^{k-2}} 2\left(1+e^{i\pi\widehat{Q{\oplus}Q'{\oplus}QQ'}}\right) e^{2i\pi R}\ketbra{\vec O}{\vec I}\\
&= 2\sum_{\vec y\in\{0,1\}^{k-1}} e^{2i\pi\left(\frac{y_0}{2}(\widehat Q + \widehat{Q'} + \widehat{QQ'}) + R\right)}\ketbra{\vec O}{\vec I}\\
&= \interp{2t[y_0'\leftarrow y_0\oplus y_0Q]}
\end{align*}
The third equality is obtained by checking that the equality is true for all values of $\widehat{Q}$ and $\widehat{Q'}$:
$$\begin{array}{c|c|c|c}
~\widehat Q~ & ~\widehat{Q'}~ & ~\left(1+e^{i\pi\widehat Q}+e^{i\pi\widehat Q'}+e^{i\pi\widehat{Q{\oplus}Q'}}\right)~ & ~2\left(1+e^{i\pi\widehat{Q{\oplus}Q'{\oplus}QQ'}}\right)\\
\hline
0 & 0 & 4 & 4 \\
0 & 1 & 0 & 0 \\
1 & 0 & 0 & 0 \\
1 & 1 & 0 & 0 
\end{array}$$

\end{proof}

\begin{proof}[Proof of \autoref{lem:H-spider-decomp}]
\phantomsection{.}\label{prf:H-spider-decomp}
First, thanks to rule (HS1), we have $\tikzfig{H-spider}\eq{}\tikzfig{H-spider-decomp-aux}$. Then, we have:
\def\fig{H-spider-decomp-proof}
\begin{align*}
\interp{\begin{tikzpicture}
	\begin{pgfonlayer}{nodelayer}
		\node [style=white dot] (0)  at (-0.5, 0.0) {};
		\node [style=white dot] (1)  at (0.5, 0.0) {};
		\node [style=white dot] (2)  at (0.0, 0.5) {};
		\node [style=box] (3)  at (-0.5, 0.5) {};
		\node [style=box] (4)  at (0.0, 0.0) {};
		\node [style=box] (5)  at (0.5, 0.5) {};
		\node [style=none] (6)  at (-0.5, -0.5) {};
		\node [style=none] (7)  at (0.5, -0.5) {};
	\end{pgfonlayer}
	\begin{pgfonlayer}{edgelayer}
		\draw (0) to (6.center);
		\draw (0) to (1);
		\draw (1) to (7.center);
		\draw (2) to (5);
		\draw (2) to (4);
		\draw (3) to (0);
		\draw (3) to (2);
		\draw (5) to (1);
	\end{pgfonlayer}
\end{tikzpicture}}=\frac12\begin{pmatrix}1+r\\1-r\end{pmatrix}=\frac{1+r}2\begin{pmatrix}1\\\frac{1-r}{1+r}\end{pmatrix}
\end{align*}
and
\begin{align*}
\interp{\begin{tikzpicture}
	\begin{pgfonlayer}{nodelayer}
		\node [style=white dot] (20)  at (-0.375, 0.013) {};
		\node [style=box] (21)  at (0.625, 0.512) {};
		\node [style=box] (22)  at (0.625, 0.013) {};
		\node [style=none] (23)  at (-0.375, -0.738) {};
		\node [style=none] (24)  at (0.625, -0.738) {};
		\node [style=dot] (25)  at (-0.375, -0.412) {};
		\node [style=none, font={\scriptsize}] (26)  at (-0.625, -0.438) {$\neg$};
		\node [style=dot] (27)  at (0.125, 0.512) {};
		\node [style=none, font={\scriptsize}] (28)  at (0.125, 0.738) {$\neg$};
	\end{pgfonlayer}
	\begin{pgfonlayer}{edgelayer}
		\draw (20) to (23.center);
		\draw (20) to (22);
		\draw (21) to (22);
		\draw (21) to (27);
		\draw (22) to (24.center);
		\draw (27) to (20);
	\end{pgfonlayer}
\end{tikzpicture}}
=2se^{i\frac\alpha2}\begin{pmatrix}\cos{\frac\alpha2}\\-ie^{i\beta}\sin{\frac\alpha2}\end{pmatrix}
=2se^{i\frac\alpha2}\cos{\frac\alpha2}\begin{pmatrix}1\\e^{i(\beta-\frac\pi2)}\tan{\frac\alpha2}\end{pmatrix}
\end{align*}
Hence, when $\abs{r}\notin\{0,1\}$, we have equality between the two with
$\alpha := 2\atan{\frac{1-r}{1+r}}$, $\beta = \arg\left(\frac{1-r}{1+r}\right)+\frac\pi2$ and $s:=\frac{1+r}{4e^{i\frac\alpha2}\cos{\frac\alpha2}}$ (since $r\neq1$, $\alpha$ is well defined and $\alpha\neq\pi\mod 2\pi$ so $s$ is also well-defined). From this, we get:
\[\interp{\tikzfig{H-spider}} = \interp{\tikzfig{H-spider-decomp}}\qedhere\]
\end{proof}

\begin{proof}[Proof of \autoref{lem:poly-copy}]
\phantomsection\label{prf:poly-copy}
\def\fig{polynomial-copy}
\begin{align*}
\begin{tikzpicture}
	\begin{pgfonlayer}{nodelayer}
		\node [style=white dot] (0)  at (-0.5, 0.0) {};
		\node [style=white dot] (1)  at (0.5, 0.0) {};
		\node [style=white dot] (2)  at (0.0, 0.5) {};
		\node [style=box] (3)  at (-0.5, 0.5) {};
		\node [style=box] (4)  at (0.0, 0.0) {};
		\node [style=box] (5)  at (0.5, 0.5) {};
		\node [style=none] (6)  at (-0.5, -0.5) {};
		\node [style=none] (7)  at (0.5, -0.5) {};
	\end{pgfonlayer}
	\begin{pgfonlayer}{edgelayer}
		\draw (0) to (6.center);
		\draw (0) to (1);
		\draw (1) to (7.center);
		\draw (2) to (5);
		\draw (2) to (4);
		\draw (3) to (0);
		\draw (3) to (2);
		\draw (5) to (1);
	\end{pgfonlayer}
\end{tikzpicture}
&\eq{}\begin{tikzpicture}
	\begin{pgfonlayer}{nodelayer}
		\node [style=white dot] (9)  at (-0.5, 0.0) {};
		\node [style=box] (12)  at (-0.5, 0.5) {};
		\node [style=box] (13)  at (0.5, 0.0) {};
		\node [style=none] (15)  at (-0.5, -0.5) {};
		\node [style=none] (16)  at (0.5, -0.5) {};
		\node [style=dot] (17)  at (0.0, 0.0) {};
		\node [style=none, font={\scriptsize}] (18)  at (0.0, -0.275) {$\neg$};
	\end{pgfonlayer}
	\begin{pgfonlayer}{edgelayer}
		\draw (9) to (15.center);
		\draw (9) to (13);
		\draw (12) to (9);
		\draw [bend left] (12) to (13);
		\draw (13) to (16.center);
	\end{pgfonlayer}
\end{tikzpicture}
\eq{}\begin{tikzpicture}
	\begin{pgfonlayer}{nodelayer}
		\node [style=white dot] (20)  at (-0.375, 0.013) {};
		\node [style=box] (21)  at (0.625, 0.512) {};
		\node [style=box] (22)  at (0.625, 0.013) {};
		\node [style=none] (23)  at (-0.375, -0.738) {};
		\node [style=none] (24)  at (0.625, -0.738) {};
		\node [style=dot] (25)  at (-0.375, -0.412) {};
		\node [style=none, font={\scriptsize}] (26)  at (-0.625, -0.438) {$\neg$};
		\node [style=dot] (27)  at (0.125, 0.512) {};
		\node [style=none, font={\scriptsize}] (28)  at (0.125, 0.738) {$\neg$};
	\end{pgfonlayer}
	\begin{pgfonlayer}{edgelayer}
		\draw (20) to (23.center);
		\draw (20) to (22);
		\draw (21) to (22);
		\draw (21) to (27);
		\draw (22) to (24.center);
		\draw (27) to (20);
	\end{pgfonlayer}
\end{tikzpicture}
\eq{}\begin{tikzpicture}
	\begin{pgfonlayer}{nodelayer}
		\node [style=white dot] (30)  at (-0.375, -0.125) {};
		\node [style=box] (31)  at (-0.125, 0.375) {};
		\node [style=box] (32)  at (0.625, -0.125) {};
		\node [style=none] (33)  at (-0.375, -0.875) {};
		\node [style=none] (34)  at (0.625, -0.875) {};
		\node [style=dot] (35)  at (-0.375, -0.55) {};
		\node [style=none, font={\scriptsize}] (36)  at (-0.625, -0.575) {$\neg$};
		\node [style=white dot] (37)  at (0.375, 0.375) {};
		\node [style=box] (38)  at (0.375, 0.875) {};
	\end{pgfonlayer}
	\begin{pgfonlayer}{edgelayer}
		\draw (30) to (33.center);
		\draw (30) to (32);
		\draw (30) to (31);
		\draw (31) to (37);
		\draw (32) to (34.center);
		\draw (37) to (32);
		\draw (38) to (37);
	\end{pgfonlayer}
\end{tikzpicture}\\
&\eq{}\begin{tikzpicture}
	\begin{pgfonlayer}{nodelayer}
		\node [style=box] (42)  at (0.125, 0.125) {};
		\node [style=none] (43)  at (-0.375, -0.625) {};
		\node [style=none] (44)  at (0.625, -0.625) {};
		\node [style=dot] (45)  at (-0.375, -0.3) {};
		\node [style=none, font={\scriptsize}] (46)  at (-0.625, -0.325) {$\neg$};
		\node [style=box] (48)  at (0.125, 0.625) {};
		\node [style=dot] (49)  at (0.625, -0.3) {};
		\node [style=none, font={\scriptsize}] (50)  at (0.375, -0.325) {$\neg$};
	\end{pgfonlayer}
	\begin{pgfonlayer}{edgelayer}
		\draw (42) to (45);
		\draw (44.center) to (49);
		\draw (45) to (43.center);
		\draw (48) to (42);
		\draw (49) to (42);
	\end{pgfonlayer}
\end{tikzpicture}
\eq{}\begin{tikzpicture}
	\begin{pgfonlayer}{nodelayer}
		\node [style=none] (53)  at (-0.25, -0.625) {};
		\node [style=none] (54)  at (0.5, -0.625) {};
		\node [style=dot] (55)  at (-0.25, -0.3) {};
		\node [style=none, font={\scriptsize}] (56)  at (-0.5, -0.325) {$\neg$};
		\node [style=dot] (58)  at (0.5, -0.3) {};
		\node [style=none, font={\scriptsize}] (59)  at (0.25, -0.325) {$\neg$};
		\node [style=dot] (60)  at (-0.25, 0.2) {};
		\node [style=none, font={\scriptsize}] (61)  at (-0.5, 0.175) {$\neg$};
		\node [style=dot] (62)  at (0.5, 0.2) {};
		\node [style=none, font={\scriptsize}] (63)  at (0.25, 0.175) {$\neg$};
		\node [style=white dot] (74)  at (0.0, 0.625) {};
	\end{pgfonlayer}
	\begin{pgfonlayer}{edgelayer}
		\draw (54.center) to (58);
		\draw (55) to (53.center);
		\draw (60) to (55);
		\draw (62) to (58);
	\end{pgfonlayer}
\end{tikzpicture}\qedhere
\end{align*}
\end{proof}

\begin{proof}[Proof of \autoref{prop:TH-proves-TH}]
\phantomsection\label{prf:TH-proves-TH}
We show that all the rules of $\operatorname{ZH}_{\operatorname{TH}}$ hold in $\cat{SOP}[\frac12]$.

Checking the rules (ZS1), (ZS2), (HS1), (HS2) and (M) is straightforward using the rule (HH). We give for instance a check of the rule (ZS1):
\begin{align*}
\left[\tikzfig{ZH-rule-ZS1-lhs}\right]^{\operatorname{sop}} &= \frac12\sum e^{2i\pi \frac{y_0+y_1}2y'}\ketbra{y_1,...,y_1}{y_0,...,y_0}\\
&\rewrite{\text{HH}(y',[y_1\leftarrow y_0])}
\sum \ketbra{y_0,...,y_0}{y_0,...,y_0} = \left[\tikzfig{Z-spider}\right]^{\operatorname{sop}}
\end{align*}
We give derivations to prove the remaining rules of $\operatorname{ZH}_{\operatorname{TH}}$. Recall that equality is up to $\alpha$-conversion.

(IV):
\begin{align*}
\left[~\tikzfig{H-scalar-1_2}~~\tikzfig{Z-0-0}~\right]^{\operatorname{sop}} = \frac12\sum_y\ketbra{}{}
\rewrite{\text{Elim}} 1 = \left[~\tikzfig{empty}~\right]^{\operatorname{sop}}
\end{align*}

(Z):
\begin{align*}
\left[~\tikzfig{H-scalar-1_sqrt2}~~\tikzfig{ZH-scalar-0}~\right]^{\operatorname{sop}}
\rewrite{\text{HH}}
\frac1{\sqrt2}\sum e^{2i\pi \frac y2}\ketbra{}{}
\rewrite{\text Z}
\sum e^{2i\pi \frac y2}\ketbra{}{}
\underset{\text{HH}}{\longleftarrow}
\left[~\tikzfig{ZH-scalar-0}~\right]^{\operatorname{sop}}
\end{align*}

The two rules (BA1) and (BA2) are fairly easy to check, once one realises that $\left[\tikzfig{xor}\right]^{\operatorname{sop}} \rewrite{\text{HH}} \sum \ketbra{y_0{\oplus}y_1}{y_0,y_1}$:
\begin{align*}
\left[\tikzfig{ZH-rule-BA1-rhs}\right]^{\operatorname{sop}} \rewrite{\text{HH}}
\frac12\sum e^{2i\pi \frac{y_1+...+y_n+y_0}2y'} \ketbra{y_1,...,y_n}{y_0,...,y_0}\\
\rewrite{\text{HH}(y',[y_0\leftarrow y_1{\oplus}...{\oplus}y_n])}
\sum \ketbra{y_1,...,y_n}{y_1{\oplus}...{\oplus}y_n,...,y_1{\oplus}...{\oplus}y_n}
\underset{\text{HH}}{\longleftarrow}\left[\tikzfig{ZH-rule-BA1-lhs}\right]^{\operatorname{sop}} 
\end{align*}
and
\begin{align*}
\left[\tikzfig{ZH-rule-BA2-rhs}\right]^{\operatorname{sop}} \rewrite{\text{HH}}
\frac12\sum e^{2i\pi \left(\frac{y_1...y_my'}2+\frac{x_1+...+x_n+y'}2y''\right)} \ketbra{y_1,...,y_m}{x_1,...x_n}\\
\rewrite{\text{HH}(y'',[y'\leftarrow x_1{\oplus}...{\oplus}x_n])}
\sum e^{2i\pi \left(\frac{y_1...y_mx_1}2+...+\frac{y_1...y_mx_n}2\right)} \ketbra{y_1,...,y_m}{x_1,...x_n}\\
\underset{\text{HH}}{\longleftarrow}\left[\tikzfig{ZH-rule-BA2-lhs}\right]^{\operatorname{sop}} 
\end{align*}

(O):
\begin{align*}
\left[\tikzfig{ZH-rule-O-lhs}\right]^{\operatorname{sop}}\rewrite{\text{HH}}
2 \sum e^{2i\pi\left(\frac{1}{2}y_{0}y_{2}y_{3} + \frac{1}{2}y_{0}y_{3} + \frac{1}{2}y_{1}y_{2}y_{3}\right)}\ket{y_{0}, y_{1}}\!\!\bra{y_{2}}
\end{align*}
\begin{align*}
\left[\tikzfig{ZH-rule-O-rhs}\right]^{\operatorname{sop}}\rewrite{\text{HH}}
\sum e^{2i\pi\left(\frac{1}{2}y_{0}y_{1} + \frac{1}{2}y_{2}y_{3}y_{4} + \frac{1}{2}y_{0}y_{1}y_{4}\right)}\ket{y_{0}, y_{3}}\!\!\bra{y_{4}}\\
\rewrite{\text{HHnl}(y_{1},y_{2})}
2 \sum e^{2i\pi\left(\frac{1}{2}y_{0}y_{1} + \frac{1}{2}y_{0}y_{1}y_{4} + \frac{1}{2}y_{1}y_{3}y_{4}\right)}\ket{y_{0}, y_{3}}\!\!\bra{y_{4}}
\end{align*}

(\&):
\begin{align*}
\left[\tikzfig{ZH-rule-AND-lhs}\right]^{\operatorname{sop}}\rewrite{\text{HH}}
\sum e^{2i\pi\left(\frac{1}{2}y_{0}y_{1} + \frac{1}{2}y_{0}y_{2}y_{3}\right)}\ket{y_{1}}\!\!\bra{y_{2}, y_{3}}
\rewrite{\text{HH}(y_{0},[y_{1}\leftarrow y_{2}y_{3}])}
2 \sum \ket{y_{2}y_{3}}\!\!\bra{y_{2}, y_{3}}
\end{align*}
\begin{align*}
\left[\tikzfig{ZH-rule-AND-rhs}\right]^{\operatorname{sop}}\rewrite{\text{HH}}
\frac{1}{4} \sum e^{2i\pi\left(\frac{1}{2}y_{0} + \frac{1}{2}y_{8}y_{1}y_{7} + \frac{1}{2}y_{1} + \frac{1}{2}y_{1}y_{3}y_{4} + \frac{1}{2}y_{0}y_{1}y_{2} + \frac{1}{2}y_{0}y_{2}\right)}\raisebox{-2ex}{\hspace*{-2em}$\ket{y_{0}}\!\!\bra{1{\oplus}y_{4}, 1{\oplus}y_{7}}$}\\
\rewrite{\text{HHnl}(y_{8},y_{2})}
\frac{1}{2} \sum e^{2i\pi\left(\frac{1}{2}y_{0}y_{1}y_{8} + \frac{1}{2}y_{0} + \frac{1}{2}y_{8}y_{1}y_{7} + \frac{1}{2}y_{1} + \frac{1}{2}y_{1}y_{3}y_{4} + \frac{1}{2}y_{0}y_{8}\right)}\ket{y_{0}}\!\!\bra{1{\oplus}y_{4}, 1{\oplus}y_{7}}\\
\rewrite{\text{ket/bra}([y_{4}\leftarrow y_{4}{\oplus}1])\\\text{ket/bra}([y_{7}\leftarrow y_{7}{\oplus}1])}
\frac{1}{2} \sum e^{2i\pi\left(\frac{1}{2}y_{0}y_{1}y_{8} + \frac{1}{2}y_{0} + \frac{1}{2}y_{8}y_{1}y_{7} + \frac{1}{2}y_{8}y_{1} + \frac{1}{2}y_{1} + \frac{1}{2}y_{1}y_{3}y_{4} + \frac{1}{2}y_{1}y_{3} + \frac{1}{2}y_{0}y_{8}\right)}\raisebox{-2ex}{\hspace*{-2em}$\ket{y_{0}}\!\!\bra{y_{4}, y_{7}}$}\\
\rewrite{\text{HHnl}(y_{8},y_{3})}
\sum e^{2i\pi\left(\frac{1}{2}y_{0}y_{1}y_{8} + \frac{1}{2}y_{0} + \frac{1}{2}y_{8}y_{1}y_{4}y_{7} + \frac{1}{2}y_{1} + \frac{1}{2}y_{8}y_{1} + \frac{1}{2}y_{0}y_{8}\right)}\ket{y_{0}}\!\!\bra{y_{4}, y_{7}}\\
\rewrite{\text{HHgen}(y_{1},[y_{0}\leftarrow y_{4}y_{7}y_8{\oplus}y_8{\oplus}1])}
\sum e^{2i\pi\left(\frac{1}{2}y_{1} + \frac{1}{2}y_{8}y_{1} + \frac{1}{2}y_{8} + \frac12\right)}\ket{y_{4}y_{7}y_8{\oplus}y_8{\oplus}1}\!\!\bra{y_{4}, y_{7}}\\
\rewrite{\text{HH}(y_1,[y_8\leftarrow 1])}
2 \sum \ket{y_{4}y_{7}}\!\!\bra{y_{4}, y_{7}}
\end{align*}
\end{proof}

\begin{proof}[Proof of \autoref{prop:ascending-encodes}]
\phantomsection\label{prf:ascending-encodes}
We demonstrate this proposition by showing that:
\begin{enumerate}
\item $\interp{\cat{SOP}[\frac1{2^{k+1}}]'}\subseteq
\mathcal M(\mathbb Z[\frac12,e^{i\frac\pi{2^k}}])$
\label{item:interp}
\item For each element $x\in \mathbb Z[\frac12,e^{i\frac\pi{2^k}}]$, there exists a unique decomposition as $x = x_1+e^{i\frac\pi{2^k}}x_2$ where $x_1,x_2\in \mathbb Z[\frac12,e^{i\frac\pi{2^{k-1}}}]$
\label{item:unique}
\item There exists a map $\psi_k:\mathcal M(\mathbb Z[\frac12, e^{i\frac\pi{2^k}}])\to \mathcal M(\mathbb Z[\frac12, e^{i\frac\pi{2^{k-1}}}])$, based on the decomposition, and such that $\interp{\ascend{t}_k} = \psi_k\left(\interp{t}\right)$
\label{item:exists}
\end{enumerate}
In this case, given $t_1,t_2\in\cat{SOP}[\frac1{2^{k+1}}]'$ such that $\interp{t_1}=\interp{t_2}$, by \ref{item:interp}.~we can apply $\psi_k$ to their interpretation. By uniqueness of the decomposition \ref{item:unique}., $\psi_k(\interp{t_1})=\psi_k(\interp{t_2})$. Finally, by \ref{item:exists}., $\interp{\ascend{t_1}_k} = \interp{\ascend{t_2}_k}$.
Let us now prove the previous claims:
\begin{enumerate}
\item This point is a simple verification.
\item Let $\displaystyle x=\sum_{\ell=0}^{2^k-1} \alpha_\ell e^{i\frac{\ell\pi}{2^k}}$ $\in \mathbb Z[\frac12,e^{i\frac\pi{2^k}}]$. Obviously, $x$ can be decomposed as $$ x = \sum_{\ell=0}^{2^{k-1}-1} \alpha_{2\ell} e^{i\frac{\ell\pi}{2^{k-1}}} + e^{i\frac{\pi}{2^k}}\sum_{\ell=0}^{2^{k-1}-1} \alpha_{2\ell+1} e^{i\frac{\ell\pi}{2^{k-1}}} = x_1+e^{i\frac\pi{2^k}}x_2$$
where $x_1,x_2\in \mathbb Z[\frac12,e^{i\frac\pi{2^{k-1}}}]$. We now need to show that this decomposition is unique. To do so, let us consider $\mathbb Q[e^{i\frac\pi{2^k}}]$ and $\mathbb Q[e^{i\frac\pi{2^{k-1}}}]$. These are two fields such that $\mathbb Q[e^{i\frac\pi{2^{k-1}}}] \subset \mathbb Q[e^{i\frac\pi{2^k}}]$. $\mathbb Q[e^{i\frac\pi{2^k}}]$ can hence be seen as a vector space over $\mathbb Q[e^{i\frac\pi{2^{k-1}}}]$. This vector space is of dimension:
$$\left[\mathbb Q[e^{i\frac\pi{2^k}}]:\mathbb Q[e^{i\frac\pi{2^{k-1}}}]\right]
= \left[\mathbb Q[e^{i\frac{2\pi}{2^{k+1}}}]:\mathbb Q[e^{i\frac{2\pi}{2^k}}]\right]
=\frac{\left[\mathbb Q[e^{i\frac{2\pi}{2^{k+1}}}]:\mathbb Q\right]}{\left[\mathbb Q[e^{i\frac{2\pi}{2^k}}]:\mathbb Q\right]} = \frac{\varphi(2^{k+1})}{\varphi(2^k)}=\frac{2^k}{2^{k-1}}=2$$
where $\varphi$ is Euler's totient function. The vector space has $(1,e^{i\frac\pi{2^k}})$ as a basis. Hence, the above decomposition is unique.
\item We now need to define $\psi_k$. We are going to define it first on scalars, and on the basis $(1,e^{i\frac\pi{2^k}})$:
$$\psi_k(1):= I_2 = \begin{pmatrix}1&0\\0&1\end{pmatrix}\qquad\quad\text{and}\qquad\quad
\psi_k(e^{i\frac\pi{2^k}}):= X_k = \begin{pmatrix}0&1\\e^{i\frac\pi{2^{k-1}}}&0\end{pmatrix}$$
By linearity, $\psi_k$ is defined on all elements of $\mathbb Z[\frac12,e^{i\frac\pi{2^k}}]$. We then naturally extend this definition to any matrix over these elements. Formally: $\psi_k:A+Be^{i\frac\pi{2^k}}\mapsto A\otimes I_2 + B\otimes X_k$ where $A+Be^{i\frac\pi{2^k}}$ is the aforementioned decomposition extended to matrices. One can check that $\psi_k$ is a homomorphism, i.e.~$\psi_k(.+.)=\psi_k(.)+\psi_k(.)$ and $\psi_k(.\circ .)=\psi_k(.)\circ\psi_k(.)$.

It remains to show that $\interp{\ascend{.}_k} = \psi_k\left(\interp{.}\right)$. Since $\psi_k$ is a homomorphism, it is enough to show the result on the terms in the decomposed form of \autoref{lem:sop-decomp}. Let $t=s\sum e^{2i\pi \frac\ell{2^{k+1}}y_{i_1}...y_{i_q}} \ketbra{\vec O}{\vec I}$ be such a term.

If $\ell\bmod2=0$, then $\interp{t}\in\mathcal M(\mathbb Z[\frac12, e^{i\frac\pi{2^{k-1}}}])$ so $\psi_k(\interp t) = \interp t \otimes I_2$ and: \[\interp{\ascend{t}_k} = \interp{s\sum e^{2i\pi \frac{\ell/2}{2^{k}}y_{i_1}...y_{i_q}} \ketbra{\vec O,y'}{\vec I,y'}} = \interp t\otimes I_2.\]

If $\ell\bmod2=1$, then:
$$\interp{t} = se^{i\frac\pi{2^k}}\sum_{y_{i_1}...y_{i_q}=1} e^{2i\pi \frac{(\ell-1)/2}{2^k}} \ketbra{\vec O}{\vec I} + s\sum_{y_{i_1}...y_{i_q}=0} \ketbra{\vec O}{\vec I}$$
so:
$$\psi_k(\interp t) = \left(s\sum_{y_{i_1}...y_{i_q}=1} e^{2i\pi \frac{(\ell-1)/2}{2^k}} \ketbra{\vec O}{\vec I}\right)\otimes X_k + \left(s\sum_{y_{i_1}...y_{i_q}=0} \ketbra{\vec O}{\vec I}\right)\otimes I_2$$
and 
\begin{align*}
&\interp{\ascend{t}_k} = s\sum e^{2i\pi \frac{y_{i_1}...y_{i_q}}{2^{k}}\left((\ell-1)/2+y'\right)} \ketbra{\vec O,y'}{\vec I,y'{\oplus}y_{i_1}...y_{i_q}}\\
&= s\sum_{y_{i_1}...y_{i_q}=1} e^{2i\pi \frac{(\ell-1)/2+y'}{2^{k}}} \ketbra{\vec O,y'}{\vec I,y'{\oplus}1} + s\sum_{y_{i_1}...y_{i_q}=0} \ketbra{\vec O,y'}{\vec I,y'}\\
&= \left(s\sum_{y_{i_1}...y_{i_q}=1} e^{2i\pi \frac{(\ell-1)/2}{2^k}} \ketbra{\vec O}{\vec I}\right)\otimes X_k + \left(s\sum_{y_{i_1}...y_{i_q}=0} \ketbra{\vec O}{\vec I}\right)\otimes I_2 = \psi_k(\interp t)\qedhere
\end{align*}
\end{enumerate}
\end{proof}

\begin{proof}[Proof of \autoref{prop:descending-reverses-ascending}]
\phantomsection\label{prf:descending-reverses-ascending}
Again, we can use the decomposition given in \autoref{lem:sop-decomp}. We can show that if $t=s\sum e^{2i\pi \frac\ell{2^{k+1}}y_{i_1}...y_{i_q}} \ketbra{\vec O}{\vec I}$, then $\ascend{t}_k\circ (id_n\otimes \sum e^{2i\pi \frac{y_0}{2^{k+1}}}\ket{y_0}) \underset{\textnormal{TH}}\sim t\otimes \sum e^{2i\pi \frac{y_0}{2^{k+1}}}\ket{y_0}$:

If $\ell\bmod2=0$, then $\ascend{t}_k = t\otimes id$ so $\ascend{t}_k\circ (id_n\otimes \sum e^{2i\pi \frac{y_0}{2^{k+1}}}\ket{y_0}) \underset{\textnormal{TH}}\sim t\otimes \sum e^{2i\pi \frac{y_0}{2^{k+1}}}\ket{y_0}$.

If $\ell\bmod2=1$, then:
\begin{align*}
\ascend{t}_k\circ (id_n\otimes \sum e^{2i\pi \frac{y_0}{2^{k+1}}}\ket{y_0})
\hspace*{25em}\\
= \frac s2\sum e^{2i\pi \left( \frac{y_{i_1}...y_{i_q}}{2^{k}}\left((\ell-1)/2+y'\right)+\frac{y'+y_{i_1}...y_{i_q}+y_0}2y''+\frac{y_0}{2^{k+1}}\right)} \ketbra{\vec O,y'}{\vec I}\\
\rewrite{\text{HH}(y'',[y_0\leftarrow y'{\oplus}y_{i_1}...y_{i_q}])}
s\sum e^{2i\pi \left( \frac{y_{i_1}...y_{i_q}}{2^{k}}\left((\ell-1)/2+y'\right)+\frac{y'+y_{i_1}...y_{i_q}-2y'y_{i_1}...y_{i_q}}{2^{k+1}}\right)} \ketbra{\vec O,y'}{\vec I}\\
=s\sum e^{2i\pi \left( \ell\frac{y_{i_1}...y_{i_q}}{2^{k}}+\frac{y'}{2^{k+1}}\right)} \ketbra{\vec O,y'}{\vec I} = t\otimes \sum e^{2i\pi \frac{y_0}{2^{k+1}}}\ket{y_0}
\end{align*}

Now, for an arbitrary $t\in \cat{SOP}[\frac1{2^{k+1}}]'$, we can do the above inductively on each term in its decomposition, resulting in $\ascend{t}_k\circ (id_n\otimes \sum e^{2i\pi \frac{y_0}{2^{k+1}}}\ket{y_0}) \underset{\textnormal{TH}}\sim t\otimes \sum e^{2i\pi \frac{y_0}{2^{k+1}}}\ket{y_0}$. Finally:
\begin{align*}
\descend{\ascend{t}_k}_k &= (id_m\otimes \bra0)\circ \ascend{t}_k \circ (id_n\otimes \sum e^{2i\pi \frac{y_0}{2^{k+1}}}\ket{y_0})\\
&\underset{\textnormal{TH}}\sim (id_m\otimes \bra0)\circ \left(t\otimes \sum e^{2i\pi \frac{y_0}{2^{k+1}}}\ket{y_0}\right)
\underset{\textnormal{TH}}\sim t
\end{align*}

The second result in the Proposition simply comes from the fact that $\descend{t_i}$ is built by composition from $t_i$, so Proposition \ref{prop:TH-local} gives the desired result.
\end{proof}

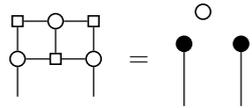
\begin{lemma}
\label{lem:HHnl-lemma}
\def\fig{HHnl-lemma}
\begin{align*}
\begin{tikzpicture}
	\begin{pgfonlayer}{nodelayer}
		\node [style=white dot] (0)  at (-0.5, 0.0) {};
		\node [style=white dot] (1)  at (0.5, 0.0) {};
		\node [style=white dot] (2)  at (0.0, 0.5) {};
		\node [style=box] (3)  at (-0.5, 0.5) {};
		\node [style=box] (4)  at (0.0, 0.0) {};
		\node [style=box] (5)  at (0.5, 0.5) {};
		\node [style=none] (6)  at (-0.5, -0.5) {};
		\node [style=none] (7)  at (0.5, -0.5) {};
	\end{pgfonlayer}
	\begin{pgfonlayer}{edgelayer}
		\draw (0) to (6.center);
		\draw (0) to (1);
		\draw (1) to (7.center);
		\draw (2) to (5);
		\draw (2) to (4);
		\draw (3) to (0);
		\draw (3) to (2);
		\draw (5) to (1);
	\end{pgfonlayer}
\end{tikzpicture}
\eq{}\begin{tikzpicture}
	\begin{pgfonlayer}{nodelayer}
		\node [style=none] (65)  at (-0.375, -0.625) {};
		\node [style=none] (66)  at (0.375, -0.625) {};
		\node [style=dot] (71)  at (-0.375, 0.2) {};
		\node [style=dot] (73)  at (0.375, 0.2) {};
		\node [style=white dot] (75)  at (-0.125, 0.625) {};
	\end{pgfonlayer}
	\begin{pgfonlayer}{edgelayer}
		\draw (71) to (65.center);
		\draw (73) to (66.center);
	\end{pgfonlayer}
\end{tikzpicture}
\end{align*}
\end{lemma}

\begin{proof} We have the following, where the numbering refers to lemmas in \cite{backens2021ZHcompleteness}:
\def\fig{HHnl-lemma}
\begin{align*}
\begin{tikzpicture}
	\begin{pgfonlayer}{nodelayer}
		\node [style=white dot] (0)  at (-0.5, 0.0) {};
		\node [style=white dot] (1)  at (0.5, 0.0) {};
		\node [style=white dot] (2)  at (0.0, 0.5) {};
		\node [style=box] (3)  at (-0.5, 0.5) {};
		\node [style=box] (4)  at (0.0, 0.0) {};
		\node [style=box] (5)  at (0.5, 0.5) {};
		\node [style=none] (6)  at (-0.5, -0.5) {};
		\node [style=none] (7)  at (0.5, -0.5) {};
	\end{pgfonlayer}
	\begin{pgfonlayer}{edgelayer}
		\draw (0) to (6.center);
		\draw (0) to (1);
		\draw (1) to (7.center);
		\draw (2) to (5);
		\draw (2) to (4);
		\draw (3) to (0);
		\draw (3) to (2);
		\draw (5) to (1);
	\end{pgfonlayer}
\end{tikzpicture}&
\eq{2.18}\begin{tikzpicture}
	\begin{pgfonlayer}{nodelayer}
		\node [style=white dot] (9)  at (-0.5, 0.0) {};
		\node [style=box] (12)  at (-0.5, 0.5) {};
		\node [style=box] (13)  at (0.5, 0.0) {};
		\node [style=none] (15)  at (-0.5, -0.5) {};
		\node [style=none] (16)  at (0.5, -0.5) {};
		\node [style=dot] (17)  at (0.0, 0.0) {};
		\node [style=none, font={\scriptsize}] (18)  at (0.0, -0.275) {$\neg$};
	\end{pgfonlayer}
	\begin{pgfonlayer}{edgelayer}
		\draw (9) to (15.center);
		\draw (9) to (13);
		\draw (12) to (9);
		\draw [bend left] (12) to (13);
		\draw (13) to (16.center);
	\end{pgfonlayer}
\end{tikzpicture}
\eq{2.19}\begin{tikzpicture}
	\begin{pgfonlayer}{nodelayer}
		\node [style=white dot] (20)  at (-0.375, 0.013) {};
		\node [style=box] (21)  at (0.625, 0.512) {};
		\node [style=box] (22)  at (0.625, 0.013) {};
		\node [style=none] (23)  at (-0.375, -0.738) {};
		\node [style=none] (24)  at (0.625, -0.738) {};
		\node [style=dot] (25)  at (-0.375, -0.412) {};
		\node [style=none, font={\scriptsize}] (26)  at (-0.625, -0.438) {$\neg$};
		\node [style=dot] (27)  at (0.125, 0.512) {};
		\node [style=none, font={\scriptsize}] (28)  at (0.125, 0.738) {$\neg$};
	\end{pgfonlayer}
	\begin{pgfonlayer}{edgelayer}
		\draw (20) to (23.center);
		\draw (20) to (22);
		\draw (21) to (22);
		\draw (21) to (27);
		\draw (22) to (24.center);
		\draw (27) to (20);
	\end{pgfonlayer}
\end{tikzpicture}
\eq{2.17}\begin{tikzpicture}
	\begin{pgfonlayer}{nodelayer}
		\node [style=white dot] (30)  at (-0.375, -0.125) {};
		\node [style=box] (31)  at (-0.125, 0.375) {};
		\node [style=box] (32)  at (0.625, -0.125) {};
		\node [style=none] (33)  at (-0.375, -0.875) {};
		\node [style=none] (34)  at (0.625, -0.875) {};
		\node [style=dot] (35)  at (-0.375, -0.55) {};
		\node [style=none, font={\scriptsize}] (36)  at (-0.625, -0.575) {$\neg$};
		\node [style=white dot] (37)  at (0.375, 0.375) {};
		\node [style=box] (38)  at (0.375, 0.875) {};
	\end{pgfonlayer}
	\begin{pgfonlayer}{edgelayer}
		\draw (30) to (33.center);
		\draw (30) to (32);
		\draw (30) to (31);
		\draw (31) to (37);
		\draw (32) to (34.center);
		\draw (37) to (32);
		\draw (38) to (37);
	\end{pgfonlayer}
\end{tikzpicture}\\
&\eq{2.18}\begin{tikzpicture}
	\begin{pgfonlayer}{nodelayer}
		\node [style=box] (42)  at (0.125, 0.125) {};
		\node [style=none] (43)  at (-0.375, -0.625) {};
		\node [style=none] (44)  at (0.625, -0.625) {};
		\node [style=dot] (45)  at (-0.375, -0.3) {};
		\node [style=none, font={\scriptsize}] (46)  at (-0.625, -0.325) {$\neg$};
		\node [style=box] (48)  at (0.125, 0.625) {};
		\node [style=dot] (49)  at (0.625, -0.3) {};
		\node [style=none, font={\scriptsize}] (50)  at (0.375, -0.325) {$\neg$};
	\end{pgfonlayer}
	\begin{pgfonlayer}{edgelayer}
		\draw (42) to (45);
		\draw (44.center) to (49);
		\draw (45) to (43.center);
		\draw (48) to (42);
		\draw (49) to (42);
	\end{pgfonlayer}
\end{tikzpicture}
\eq{2.28}\begin{tikzpicture}
	\begin{pgfonlayer}{nodelayer}
		\node [style=none] (53)  at (-0.25, -0.625) {};
		\node [style=none] (54)  at (0.5, -0.625) {};
		\node [style=dot] (55)  at (-0.25, -0.3) {};
		\node [style=none, font={\scriptsize}] (56)  at (-0.5, -0.325) {$\neg$};
		\node [style=dot] (58)  at (0.5, -0.3) {};
		\node [style=none, font={\scriptsize}] (59)  at (0.25, -0.325) {$\neg$};
		\node [style=dot] (60)  at (-0.25, 0.2) {};
		\node [style=none, font={\scriptsize}] (61)  at (-0.5, 0.175) {$\neg$};
		\node [style=dot] (62)  at (0.5, 0.2) {};
		\node [style=none, font={\scriptsize}] (63)  at (0.25, 0.175) {$\neg$};
		\node [style=white dot] (74)  at (0.0, 0.625) {};
	\end{pgfonlayer}
	\begin{pgfonlayer}{edgelayer}
		\draw (54.center) to (58);
		\draw (55) to (53.center);
		\draw (60) to (55);
		\draw (62) to (58);
	\end{pgfonlayer}
\end{tikzpicture}
\eq{2.12}\begin{tikzpicture}
	\begin{pgfonlayer}{nodelayer}
		\node [style=none] (65)  at (-0.375, -0.625) {};
		\node [style=none] (66)  at (0.375, -0.625) {};
		\node [style=dot] (71)  at (-0.375, 0.2) {};
		\node [style=dot] (73)  at (0.375, 0.2) {};
		\node [style=white dot] (75)  at (-0.125, 0.625) {};
	\end{pgfonlayer}
	\begin{pgfonlayer}{edgelayer}
		\draw (71) to (65.center);
		\draw (73) to (66.center);
	\end{pgfonlayer}
\end{tikzpicture}\qedhere
\end{align*}
\end{proof}

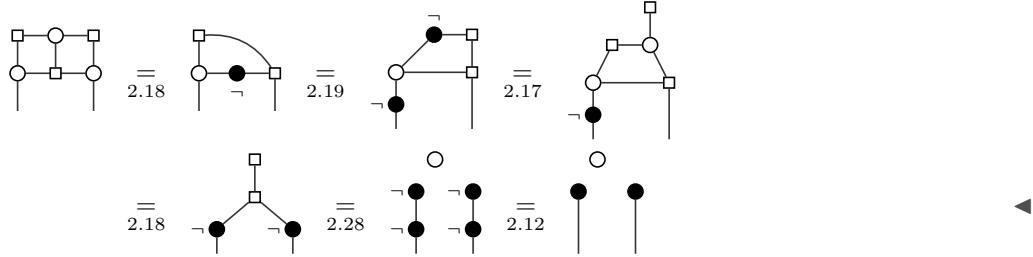
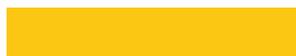
\end{document}